\documentclass{article}
\usepackage[utf8]{inputenc}
\usepackage{fullpage}
\usepackage{array}
\usepackage{makecell}

\usepackage{xcolor,colortbl}

\usepackage[
  colorlinks,
  linkcolor = blue,
  citecolor = blue,
  urlcolor = blue]{hyperref}

\usepackage{graphicx}
\newcommand{\poly}{\textnormal{poly}}
\usepackage{amsthm, amsmath, mathtools, amsfonts, enumerate}
\usepackage[scr=rsfs,scrscaled=1]{mathalpha}
\usepackage{pythonhighlight}
\usepackage[T1]{fontenc}
\usepackage{tikz}
\usetikzlibrary{quantikz}

\theoremstyle{plain}
\newtheorem{theorem}{Theorem}[section]
\newtheorem{lemma}[theorem]{Lemma}
\newtheorem{corollary}[theorem]{Corollary}
\newtheorem{conjecture}[theorem]{Conjecture}
\newtheorem{definition}[theorem]{Definition}
\newtheorem{prop}[theorem]{Proposition}

\newtheorem{remark}[theorem]{Remark}

\theoremstyle{definition}

\usepackage{cleveref}
\crefname{theorem}{Theorem}{Theorems}
\crefname{lemma}{Lemma}{Lemmas}
\crefname{proposition}{Proposition}{Propositions}
\crefname{definition}{Definition}{Definitions}
\crefname{corollary}{Corollary}{Corollaries}
\crefname{example}{Example}{Examples}
\crefname{section}{Section}{Sections}
\crefname{appendix}{Appendix}{Appendices}
\crefname{table}{Table}{Tables}
\crefname{conjecture}{Conjecture}{Conjectures}

\DeclareMathOperator{\tr}{tr}

\usepackage{xcolor}

\usepackage{dsfont} \usepackage{mathrsfs} \usepackage{microtype}

\usepackage{thmtools,thm-restate}

\renewcommand{\P}[2][]{{\textnormal{Pr}_{#1}}{\left[#2\right]}}
\newcommand{\E}[2][]{{\textnormal{E}_{#1}}{\left[#2\right]}}

\newcommand{\Proj}[1]{\textnormal{Proj}\pr*{#1}}

\newcommand{\ketbra}[2]{\ket{#1}\bra{#2}}

\newcommand{\R}{\mathbb{R}}
\newcommand{\C}{\mathbb{C}}
\newcommand{\Z}{\mathbb{Z}}

\newcommand{\eqdef}{:=}

\newcommand{\one}{\mathds{1}}

\newcommand{\Stab}{\textnormal{Stab}}

\usepackage{mathtools}
\mathtoolsset{centercolon}
\DeclarePairedDelimiter{\set}{\lbrace}{\rbrace}
\DeclarePairedDelimiter{\abs}{\lvert}{\rvert}
 \DeclarePairedDelimiter{\norm}{\lVert}{\rVert}

\DeclarePairedDelimiter{\pr}{\lparen}{\rparen}

\newcommand{\x}{\otimes}

\newcommand{\rank}{\textnormal{rank}}

\usepackage{authblk}

\newcommand{\F}{\mathbb{F}_2}

\title{Quadratic Lower bounds on the Approximate Stabilizer Rank:\\
A Probabilistic Approach}
\author{Saeed Mehraban \footnote{Tufts CS,  Medford MA, Saeed.Mehraban@tufts.edu}\quad Mehrdad Tahmasbi\footnote{UIUC CS, Champaign IL, mehrdad@illinois.edu}}
\date{\today}

\begin{document}

\maketitle
\abstract{
    The approximate stabilizer rank of a quantum state is the minimum number of terms in any approximate decomposition of that state into stabilizer states. Bravyi and Gosset \cite{Bravyi_Gosset_2016} showed that the approximate stabilizer rank of a so-called ``magic'' state like $\ket{T}^{\otimes n}$, up to polynomial factors, is an upper bound on the number of classical operations required to simulate an arbitrary quantum circuit with Clifford gates and $n$ number of  $T$ gates. As a result, an exponential lower bound on this quantity seems inevitable. Despite this intuition, several attempts using various techniques could not lead to a better than a linear lower bound on the ``exact'' rank of $\ket{T}^{\otimes n}$, meaning the minimal size of a decomposition that \emph{exactly} produces the state. 
    For the ``approximate'' rank, which is more realistically related to the cost of simulating quantum circuits,
    no lower bound better than $\tilde \Omega(\sqrt n)$ has been known
\cite{Peleg2022lowerbounds,Lovitz2022newtechniques}. In this paper, we improve the lower bound on the approximate rank to $\tilde \Omega (n^2)$ for a wide range of the approximation parameters. An immediate corollary of our result is the existence of polynomial time computable functions which require a super-linear number of terms in any decomposition into exponentials of quadratic forms over $\F$, resolving a question of \cite{williams2018limits}. Our approach is based on a strong lower bound on the approximate rank of a quantum state sampled from the Haar measure, a step-by-step analysis of the approximate rank of a magic-state teleportation protocol to sample from the Haar measure, and a result about trading Clifford operations with $T$ gates by \cite{low2018trading}.}

\section{Introduction}

Is there an efficient classical algorithm to simulate arbitrary quantum physical systems? This fundamental question plays a vital role in numerous science and engineering disciplines. For instance, in quantum chemistry, one may translate this question into the ability to measure the structural properties of molecules or design new materials \cite{chemistry}. Alternatively, in condensed matter theory, it is pertinent to our ability to predict the phases of quantum materials or sampling from thermal distributions \cite{Bauer_Bravyi_Motta_Chan_2020}. Interestingly, this question also plays a non-trivial role in seemingly unrelated fields, such as theoretical computer science, cryptography, or number theory \cite{shor1999polynomial, regev2009lattices, Jozsa2001, search}. 

In theoretical computer science, this question is formulated as the relationship between two complexity classes known as ``bounded-error quantum polynomial time'' (\textbf{BQP}) and ``bounded-error classical polynomial time'' (\textbf{BPP}). Since the early days of quantum mechanics, it was observed that many-body quantum systems have exponentially large phase spaces with counter-intuitive dynamics and unexpected features such as the duality of wave and particle aspects of subatomic systems (see \cite{weinberg2015lectures} for some of the historical remarks). They furthermore demonstrate non-classical correlations known as entanglement \cite{Einstein_Podolsky_Rosen_1935}. Hence, the popular belief is that simulation of quantum systems requires exponential classical resources. This observation indeed motivated the development of quantum computing initiated by researcher such as Feynman in the 1980s \cite{Feynman_1982}. Subsequently, a breakthrough result by Shor demonstrated that an efficient classical algorithm to simulate arbitrary quantum computations would also efficiently factor large numbers. This problem is crucial to the security of encryption schemes like RSA, for which no polynomial-time classical algorithm has been discovered despite centuries of research. It is thus natural to conjecture that $\textbf{BPP} \neq \textbf{BQP}$. However, rigorous proof for this statement seems unlikely with current complexity theoretic tools since, for example, a proof of this statement will readily yield a separation of complexity classes such as ``polynomial time'' \textbf{P} and ``polynomial space'' \textbf{PSPACE}, which has been open beside an extensive amount of research over the past few decades.

Since an unconditional separation of $\textbf{BPP}$ and $\textbf{BQP}$ seems out-of-reach, it is insightful to ask the question from a complementary angle: are there restricted but non-trivial family of quantum systems that can be efficiently simulated classical computers. It turns out that there are several examples of classically simulable quantum circuits. For instance, if the amount of quantum entanglement at every step of a quantum circuit is limited, the system can be simulated efficiently \cite{Vidal}. Examples of such circuits include log-depth one-dimensional quantum circuits. Other classically simulable circuits based on constrained architecture include constant depth 2D random quantum circuits on two-dimensional architectures \cite{Napp} or adiabatic computations with large spectral gaps \cite{Osborne}. There is, however,  an important class of quantum circuits known as Clifford circuits, which can generate maximal entanglement and have large circuit depth but admits efficient classical simulations through a well-known theorem due to Gottesmann and Knill \cite{GottesmanKnill}. This theorem works by considering quantum states ``stabilized'' by particular subgroups of the Pauli group; these states, also known as stabilizer states, were first introduced in the context of error-correcting codes \cite{calderbank1996good} and subsequently found applications in several areas in quantum information science \cite{raussendorf2000quantum, Knill_Benchmark2008}. The Clifford operations normalize the Pauli group, meaning they keep the group unchanged under conjugation, and hence, they map stabilizer states to stabilizer states. Starting from a stabilizer state, Gottesman-Knill's algorithm keeps track of generators of the stabilizer subgroup corresponding to the quantum state under simulation. This method leads to a strong simulation of Clifford circuits acting on stabilizer states on a classical computer, meaning that the output amplitudes of such computations can be computed exactly. Note that several classically simulable gate sets, such as match gates or the infinite-dimensional Gaussian gates, generate large amounts of entanglement and can be simulated classically efficiently. These simulation algorithms, however, capture variants of the Gottesman-Knill theorem for different physical particles known as Bosons \cite{CVsim} or Fermions \cite{matchgates}. In other words, the idea behind the Gottesman-Knill theorem is fundamental to all of these simulation techniques.

The Clifford gate set can only approximate a restricted family of quantum operations. However, this gate set becomes universal if we augment this gate set with an additional non-Clifford gate known as $T$, or $\pi/4$ phase shift. In the context of this universal gate set, the mentioned \textbf{BPP}  vs. \textbf{BQP} question translates to characterizing the cost of classical simulation in terms of the number of $T$ gates in an arbitrary quantum computation compiled into Clifford and $T$ gates. Building on the Gottesman and Knill theorem, several works \cite{Bravyi_Gosset_2016, Bravyi_Smith_Smolin_2016} provided upper bounds on this classical simulation cost by developing classical simulation algorithms for quantum circuits dominated by Clifford gates. In particular, Bravi and Gosset \cite{Bravyi_Gosset_2016} demonstrated that for a constant $c < 1$, a classical computer could simulate a quantum circuit with $n$ qubits, $\poly(n)$ Clifford gates, and $m$ number of $T$ gates in time poly$(n) \cdot 2^{cm}$. In their approach, they first consider teleportation of the so-called (non-stabilizer) magic state $\ket T = \frac 1 {\sqrt 2} (\ket 0 + e^{i\pi/4} \ket 1)$ to simulate the effect of a $T$ gate at the middle of the circuit using Clifford gates and measurement in the computational basis. This operation converts a quantum circuit consisting of Clifford and $T$ gates to a Clifford circuit with measurements and inputs $\ket T^{\x m}$ followed by some zero states.
Then, they decompose $\ket T^{\otimes m}$ as a linear combination of $2^{c \cdot m}$ stabilizer states and use the Gottesman-Knill algorithm for each stabilizer state in this linear combination in time $\poly(n)$. By the linearity of quantum mechanics, they can find the output of the circuit when the input is $\ket T^{\x m}$ in time $\poly(n) \cdot 2^{c \cdot m}$. In this approach, a crucial quantity is the minimum number $r$ such that there exist stabilizer states $\ket {s_1}, \ldots, \ket{s_r}$ and complex numbers $c_1, \ldots, c_r$ such that $\ket{T}^{\otimes m} = c_1 \ket{s_1} + \ldots + c_r \ket {s_r}$. We call this quantity the stabilizer rank of $\ket T ^{\x m}$ and denote it by $\chi(\ket T^{\x m})$. If we are interested in the minimum number of stabilizer states that approximate the $\ket T^{\otimes m}$ state within $\delta$ 2-norm, we arrive at the definition of the approximate rank, which we denote by $\chi_\delta(\ket T^{\x m})$. For an arbitrary state $\ket \psi$, we denote its approximate stabilizer rank with the approximation parameter $0 \leq \delta \leq 1$ by $\chi_\delta(\ket \psi)$.

In this paper, we study lower bounds on the number of steps the above simulation technique based on decomposing magic states into stabilizer states requires. In other words, our goal is to prove a lower bound on the stabilizer rank for the magic state $\ket T^{\otimes m}$. This question is essential in many ways. First, it gives significant insight into the relationship between \textbf{BQP} and \textbf{BPP} and why quantum computations obtain speedup over classical computations by studying a lower bound on \textbf{BQP} against a canonical class of simulation method on \textbf{BPP}. Similar questions have arisen in computational complexity, for instance, in the context of the \textbf{P} v.s. \textbf{NP} question,  we know that specific restricted subclasses of \textbf{P} known as monotone circuits require an exponential lower bound to solve \textbf{NP}-hard problems \cite{monot1, monot2, monot3}. 

Secondly, this question is conceptually an intriguing one. The complexity of classical simulation for quantum circuits is about counting the minimal number of computational steps that successfully simulate an ``exponentially-sized" family of problems. In contrast, the problem of computing the stabilizer rank is ``one counting problem'' about ``one'' functional structure. It is counter-intuitive that the latter, which can be viewed as a question in functional analysis, would give non-trivial information about the former. Third, as we will highlight in Section~\ref{sec:cond}, a lower bound on the (approximate) stabilizer rank organically connects with several interesting structural questions about complexity classes. For instance, we can show that if the exact rank is $r$, then $\textbf{P}^{\#\textbf{P}} \subseteq \textnormal{TIME}(\poly(n) \cdot r)/r$; here $\textnormal{TIME}(\poly(n) \cdot r)/r$ means a computation which runs in time $\poly(n) \cdot r$ and has access to $O(r)$ bits of advice providing the description of the stabilizer decomposition. One of the immediate implications of the result of Bravyi and Gosset \cite{Bravyi_Gosset_2016} is that assuming a polynomial upper bound on the approximate rank with approximation parameter $\delta$ implies that sampling within total variation distance $O(\delta)$ from arbitrary quantum circuits can be done in \textbf{BPP}; recent progress in quantum complexity theory has demonstrated that assuming plausible conjectures about the average-case hardness of specific approximate counting problems, sampling within total variation distance from quantum computers using \textbf{BPP} implies the collapse of the polynomial hierarchy (see for example \cite{AA10, bremner2011classical, bouland2017complexity}). 


\subsection{Overview of the main results}

Even though we expect the exact stabilizer rank $\chi(\ket T^{\x m})$ to grow exponentially with $m$, the best-known lower bound has been $\tilde \Omega(m)$ due to three different groups \cite{Labib2022stabilizerrank, Peleg2022lowerbounds, Lovitz2022newtechniques}. As we will explain in Section \ref{sec:prev}, these three results use three different proof techniques, but they all stop at a linear lower bound. The situation with approximate stabilizer rank is slightly worse because the best-known lower bound for this quantity is $\sqrt m$ up to poly-logarithmic factors. An immediate conjecture is whether we can prove a super-linear lower bound on either the exact or approximate ranks. 

In this paper, we resolve this conjecture by proving a nearly quadratic lower bound for the $\ket T^{\otimes m}$ state.
\begin{theorem}[Informal statement of the main result]
Let $0 <\delta < 1$, then $\chi_\delta (\ket T^{\otimes m}) = \frac{\Omega(m^2)}{\poly
\log m}$. 
\label{thm:main-informal}
\end{theorem}
As the above theorem indicates, our result works for a wide range of error parameters. The $T$ state a magic state in the second level of the Clifford hierarchy, meaning the group of operators that preserve the Clifford under conjugation; the third level is the group of operations that preserve the second level, and so on. Obviously, our result holds for quantum states that are Clifford equivalent to the $T$ state but does not hold for arbitrary magic states. We remark that for $\delta = 0$, our result holds for any magic state.

\begin{proof}[Proof Sketch.]
    Our method is a probabilistic one and has three main steps. As the first step, we show that for a random quantum state $\ket \phi$ with $n$ qubits sampled from the Haar measure, the approximate rank satisfies the following strong concentration bound
\begin{equation}
Pr [\chi_\delta (\ket \phi) \leq M] \leq e^{n^2 M - \Omega(2^n)}
\label{eq:conc}
\end{equation}
for any $\delta < 1$. As a result, we conclude that for all quantum states except for an exponentially small measure, the rank is at least $2^{n- o(1)}$. 

Second, we show we can sample from the Haar measure with exponentially small error using $\ket T^{\otimes m}$ for $m \approx n 2^{n/2}$ and $m$ adaptive measurement. We know that arbitrary quantum states can be implemented using $\tilde {O} (2^n)$ number of $T$ gates. However, by adding extra ancilla states initially in zeros and trading $T$ gates with Clifford operations, we can implement arbitrary quantum states using $\tilde O (2^{n/2})$; this result is due to \cite{low2018trading}. The third step is to show that the adaptive measurements do not increase the approximate rank of the $T$ states. We remark that this step of the analysis, in its current form, relies critically on the structure of the state and works only for the $T$ state due to its balanced structure, i.e., $|\langle 0| T\rangle|^2 = |\langle 1| T\rangle|^2$, and not arbitrary magic states. Putting these three steps together, and by rescaling $m = n 2^{n/2}$, we obtain $\frac{\Omega(m^2)}{\text{polylog}(m)}$ lower bound on the approximate rank of $\ket T^{\otimes m}$. 

The main bottleneck for going beyond the quadratic lower bound is that we need at least $\tilde {\Omega} (2^{n/2})$ $T$ states to sample with high precision from the Haar measure; see \cite{low2018trading} to see why this lower bound holds. We may wonder if we can use, instead of the Haar measure, pseudo-random quantum states such as approximate $t$-designs, which approximate the first $t$ moments of the Haar measure and use $\ll 2^{n/2}$ $T$ gates. It turns out the bound in Equation \ref{eq:conc} relies on almost all moments of the Haar measure. For instance, the main strength of this bound is due to the $2^n$ factor in the tail. For $t$-designs we only get tails like $e^{n^2 M - \Omega(t)}$. 

\end{proof}

Next, we study the relationship between circuit complexity and approximate stabilizer rank. As shown above, we obtain a quadratic lower bound on the approximate rank of a simple state like $\ket T^{\otimes m}$, which has linear circuit complexity. More generally, our result implies the following
\begin{theorem}[Stabilizer rank and circuit complexity]
For any number $d$ there exists a quantum state with circuit complexity at most $n^d \poly\log (n)$ and stabilizer rank at least $n^d$.
\label{thm:circ-complx}
\end{theorem}

We note that based on \cite{Harrow_Horodecki_2016, Haferkamp2022randomquantum} except for an exponentially small fraction, almost all quantum states from random quantum circuits of size $s$ may not have circuit complexity less than $s^{1/5}$. This gives insight into why proving an exponential lower bound on the stabilizer rank of $\ket T^{\otimes n}$ might be a difficult task; likely $\ket T^{\otimes n}$ appears as one of the rare states whose circuit complexity may be compressed and the probabilistic method may not work anymore. As we will show in Section \ref{sec:designs} a weaker version of this result can be deduced from the properties of $t$ designs. In particular, we show that a quantum state of circuit complexity at most $O(n^{5d})$ and stabilizer rank at least $n^d$ exists. However, if the dependency on the number of $T$ gates used in \cite{Roth_2023} is improved to linear, we obtain exactly the result of Theorem \ref{thm:circ-complx}. Alternatively, this result can be obtained from constructions of unitary $t$-designs based on random quantum circuits in time $O(t^{5 + o(1)}) \text{poly} (n)$ in \cite{Harrow_Horodecki_2016, Haferkamp2022randomquantum}. As indicated in \cite{hunter2019unitary}, the bound on $t$ can likely be improved to linear. In that case, we again obtain Theorem~\ref{thm:circ-complx}.

We raise the following conjecture:
\begin{conjecture}[Stabilizer rank and circuit complexity]
For any constants $d \leq d'$, there exists a quantum state with circuit complexity at most $n^d$ and stabilizer rank at least $n^{d'}$.
\label{conj:stab-complx}
\end{conjecture}

\subsection{Complexity theoretic implications}
\label{sec:quad-uncertainty}

While proving unconditional separations between complexity classes is difficult, a simpler milestone is finding complexity-theoretic lower bounds against specific families of simple functions. For instance, we can consider the problem of representing Boolean functions in a specific complexity class such as \textbf{NP} as a linear combination of simple functions.
For example, the so-called quadratic uncertainty principle \cite{filmus2014real} is the conjecture that in any exact decomposition of the AND function into 
\begin{equation}
    \sum_{j =1}^r c_j (-1)^{Q_j(x_1, \ldots, x_n)}
\label{eq:quad}
\end{equation}
we need $r \geq 2^{\Omega (n)}$, where $Q_j$ are quadratic polynomials over $\F$, and $c_j \in \mathbb C$. The best lower bound on this problem is linear due to Williams \cite{williams2018limits}. As observed by \cite{Peleg2022lowerbounds} Eq.~\eqref{eq:quad} is exactly the overlap of sums of stabilizer states with bit strings. Because the AND function is itself a stabilizer function, it has a stabilizer rank of $1$, and we do not hope to improve this bound based on a stabilizer lower bound approach. Williams furthermore showed that for any $k > 0$ there exists a function $f_k \in $\textbf{NP} such that $r \geq \Omega (n^k)$ for any decomposition of $f_k$ into linear combination $\sum_{i=1}^r\alpha_i(-1)^{Q_i}$ 
 for quadratic polynomials $Q_1, \cdots, Q_r$. An open problem is whether the same is true for functions in \textbf{P}; specifically, it is an open question whether there exists a function in \textbf{P} with super-linear $r$. As an immediate corollary of our result, we resolve this question by proving a (nearly) quadratic lower bound on the ``approximate'' rank $r$ for the function defined in the following corollary  (We remark that while the work of Williams \cite{williams2018limits} concerns exact decompositions, Chen and Williams \cite{chen2019stronger} managed to extend some of the results in \cite{williams2018limits} to hold for the approximate decompositions as well).
\begin{theorem}
 Let $M : \{0,1\}^n \rightarrow \{0,1\}$ be such that $M(x_1, \ldots, x_n) = 1$ iff  $x_1 + \ldots + x_n = 0 \mod 8$. Then
 any ``approximate'' decomposition of $M$ into a decomposition $S_r(x)$ like Eq. \eqref{eq:quad} such that $\frac 1 {2^n} \sum_x |M(x) - S_r(x)|^2 \leq \delta$ requires $r = \tilde \Omega (n^2)$.
\end{theorem}

\begin{proof}
    Our proof is inspired by an observation made in \cite{Peleg2022lowerbounds}, which implied a ``linear'' lower bound on the ``exact'' decomposition of $M$ into quadratic phases. We use our main result (Theorem \ref{thm:main-informal}) to prove a (nearly) ``quadratic'' lower bound on the number of terms in any ``approximate'' decompositions of the Boolean function $M$ into quadratic phases.
    
    Define $T(x_1, \ldots, x_n) = 2^{n/2}\langle x_1 \ldots x_n \ket T = e^{i 2\pi (\abs{x})/8 }$, where $\abs{x}$ is the Hamming weight of $x$. Let $\mathcal Q \subseteq \Stab_n$ be the set of all stabilizer states of the form $\frac 1 {\sqrt {2^n}}\sum_x (-1)^{Q(x)}\ket x$ for a quadratic polynomial $Q : \F^n \rightarrow \F$. 
    With slight modification, the (nearly) quadratic lower bound in Theorem~\ref{thm:circ-complx} also applies when we define the notion of rank with respect to the set $\mathcal Q$ instead of $\Stab_n$. As a result, for any string $x \in \F^n$ and for any decomposition satisfying $\frac 1 {2^n}\sum_x|T(x) - \sum_{j = 1}^{r_n}(-1)^{Q(x)}|^2 \leq \delta$, we have $r_n = \tilde \Omega (n^2)$, where $\delta$ is a constant fixed in advance. $T$ depends only on $|x| \mod 8$ and therefore can be written as $T(x) = \sum_{j = 0}^7 e^{2\pi i  \frac {j}8} M_j (x)$ such that $M_j (x) = 1$ iff $\abs{x} = j \mod 8$. Therefore, using the triangle inequality, there exists $j \in \{0, 1, \ldots, 7\}$ such that any decomposition $S_r$ with $r$ terms that satisfies $\frac 1 {2^n} \sum_x|M_j(x) - S_r(x)|^2 \leq \delta/7$ implies $r = \tilde \Omega (n^2)$. If $j = 0$, we are done. If $j \neq 0$ we use the following reduction: 
    $$
    M_0 (\underbrace{1,\ldots,1}_j, x_1, \ldots, x_n) = M_j (x_1, \ldots, x_n).
    $$
    Suppose there exists a decomposition with $r$ terms $S_r$ such that $\frac 1 {2^{n + j}} \sum_{y,x} |M_0 (y_1, \ldots, y_j, x_1, \ldots, x_n) - S (y_1, \ldots, y_j, x_1, \ldots, x_n )|^2 \leq \lambda$ therefore $S'_r (x_1, \ldots, x_n)= S (y_1 = 1,\ldots, y_j = 1, x_1, \ldots, x_n)$ which has $\leq r$ terms satisfies $\frac 1 {2^n} \sum_x |M_j(x) - S'_r(x)|^2 \leq 2^j \lambda$. By choosing $\lambda = \delta/ (7 \cdot 2^j)$ we conclude the proof. 
\end{proof} 

\begin{remark}
    The above result can be viewed as a lower bound for a classical problem from an upper bound for a quantum problem. In particular, we use an upper bound on the number of $T$ gates that can synthesize a quantum state with specific statistical properties to give a lower bound on the number of quadratic terms required to synthesize $M(x)$. 
\end{remark}

\subsection{Conditional lower bounds on the approximate rank}
\label{sec:cond}
Proving an unconditional exponential lower bound on the approximate stabilizer rank seems to be a difficult task. Can we prove this statement assuming plausible conjectures? First, consider the exact rank. We can show the following result.

\begin{theorem}
The exact rank of the magic state $|T\rangle^{\otimes m}$ is super-polynomial unless the permanent has polynomial circuits.
\label{thm:perm-det}
\end{theorem}

\begin{proof}
We show that if $\chi(\ket T^{\x m}) = \poly(m)$, there exists a polynomial-time algorithm with polynomial advice (and therefore a polynomial-size circuit) for the problem of computing the gap of a polynomial-size classical circuit. A polynomial-size circuit for this problem implies a polynomial-size circuit for the permanent.  To see this, given a polynomial-size function $f: \{0,1\}^\ast \rightarrow \{0,1\}$, let $U_f$ be a reversible implementation of the map $\ket x \ket 0 \rightarrow \ket x \ket {f(x)}$. 
Then $\langle 0^n 1| H^{\otimes n+1} U_f H^{\otimes n} \otimes I|0^{n+1}\rangle = \frac{gap(f)}{\sqrt 2 2^n}$. 

Next, we compile this circuit into Clifford and $T$ gates. Note that Hadamard gates are already Clifford, and we can implement $U_f$ using Toffoli gates. Furthermore, a Toffoli gate can be produced exactly using four Clifford gates and two $T$ gates (See Figure 1 of \cite{amy2013meet}).
Next, we demonstrate how to compute this quantity using Clifford gates and $T$ states. If we have $m = \poly(n)$ number of $T$ gates, using \cite[Eq.~(16)]{Bravyi_2019}, we obtain
\begin{align}
    2^{m/2}\langle 0^n 1| \otimes \bra{0^m} C_f |0^{n+1}\rangle \otimes \ket{T}^{\x m}= \frac{gap(f)}{\sqrt 2 2^n}
\end{align}
for an appropriate polynomial-size Clifford circuit $C_f$. Consider a decomposition 
\begin{align}
    \ket T^{\x m} = \sum_{i=1}^{r} c_i \ket{s_i}
\end{align}
for $r=\chi(\ket T^{\x m})$ and  stabilizer states $\ket{s_1}, \cdots, \ket{s_r}$. We show that if $r$ is at most polynomially large in $m$ (and hence $n$), then we can use polynomial in $n$ number bits of advice to represent $c_1, \ldots, c_r$. It is not immediately clear that polynomial bits of advice that keep a $2^{-m^{O(1)}}$ precision for $c_i$s is sufficient; we prove this statement in the following. Let $G \in \mathbb{C}^{r \times r}$ be the Gram matrix corresponding to $\ket{s_1}, \ldots, \ket{s_r}$ with entries $G_{ij} = \langle{s_i|s_j}\rangle$. Let $\beta_i = \langle{s_i|T^{\otimes m}}\rangle$, and let $\boldsymbol{\beta}$ and $\mathbf{c}$, respectively, be $r \times 1$ vectors with $i$'th entry $\beta_i$ and $c_i$. Therefore $\boldsymbol{\beta} = G \mathbf{c}$. Next, we observe that $G$ is invertible; because, $\ket{s_1}, \ldots \ket{s_r}$ are linearly independent, otherwise, the decomposition of $\ket{T^{\otimes m}}$ would not be minimal. Therefore,  to prove that polynomial number of bits of advice is enough to store $c_i$s with $2^{-m^{O(1)}}$ precision,  it is enough to show that polynomial bits of precision are enough to store $G^{-1}$. To see this, consider a Gaussian elimination approach to compute $G^{-1}$; for each step of the algorithm, polynomial bits of precision are enough to store the intermediate steps of the algorithm. We also use the advice to store the full description of $\ket{s_1}, \ldots, \ket{s_r}$; in order to represent each stabilizer state, we need to use the advice to store the description of an affine subspace $A$ ($\leq m$ generators), a quadratic function $Q$, and a linear function $l$, which we can do with polynomial bits of precision (see Section \ref{sec:stab-form}). Note that our argument heavily relies on the assumption that $r = m^{O(1)}$ (which is the assumption we made in the statement of this Theorem).

We then compute the gap using the following expression
\begin{align}
    2^{m/2}\langle 0^n 1| \otimes \bra{0^m} C_f |0^{n+1}\rangle \otimes \ket{T}^{\x m} = \sum_{i=1}^{\chi(\ket T^{\x m})} c_i 2^{m/2}\langle 0^n 1| \otimes \bra{0^m} C_f |0^{n+1}\rangle \otimes \ket{s_i}.
\end{align}
Using the Gottesman-Knill algorithm, we can find each $\langle 0^n 1| \otimes \bra{0^m} C_f \otimes I|0^{n+1}\rangle \otimes \ket{s_i}$ in polynomial time and evaluate the whole sum using polynomial bits of advice representing $c_i$. Hence, we can find $gap(f)$ in polynomial time conditioned on $\chi(\ket T^{\x m}) = \poly(m)$.

\end{proof}

\begin{remark}
    In complexity theory literature, it is known that permanent having polynomial-size circuits is equivalent to the collapse of complexity classes $\textbf{P}^{\#\textbf{P}} \subseteq \textbf{P} /poly$ which implies $\textbf{P}^{\# \textbf{P}} = \textbf{MA}$.  The collapse to \textbf{MA} is by \cite{babai1991non}.
\end{remark}

\begin{remark}
    Sampling one bit from the output distribution of stabilizer circuits is complete for the complexity class $\oplus$\textbf{L} \cite{AaronsonGottesman}, which is the class of problems that are solvable on a non-deterministic logspace machine which the even parity of a non-deterministic path is an indication of acceptance. Hence, we suspect that the full power of $\textbf{P}$ is not necessary to prove Theorem~\ref{th:poly-stab-rank} and we might be able to replace \textbf{P} with a weaker class such as $\oplus$\textbf{L}. However, to show this, we need to show that strong simulation of stabilizer circuits is also possible in $\oplus$\textbf{L}. More realistically, we suspect a strong simulation of Clifford circuits is possible in gap\textbf{L} $= \textbf{DET}$, the same way a strong simulation of \textbf{BQP} is possible in gap\textbf{P}.
\end{remark}

We note that the above theorem implies the following immediate corollary for the approximate rank:
\begin{corollary}
Let $\delta = 1/2^{\Omega(n)}$, then the approximate rank of $\ket T^{\otimes n}$ is at least super-polynomial unless the permanent polynomial has polynomial-size circuits.
\end{corollary}

\begin{proof}
    Let $\|\ket \psi - \ket T^{\otimes n}\|_2 \leq \delta$, then $\ket \psi = \ket T^{\otimes n} + \ket \delta$ where $\|\ket \delta \|_2 = O(\delta)$. Therefore, similar to the proof of Theorem \ref{thm:perm-det}, the gap of a function $f$ can be encoded as 
    $$
    gap(f) = \sqrt 2 2^{n + m/2}\langle 0^n 1| \otimes \bra{0^m} C_f |0^{n+1}\rangle \otimes \ket \psi + O( 2^{m/2 + n}{\delta})
    $$
    Without loss of generality, we can assume $n = O(m)$ because $m$ is a constant times the size of the circuit for $f$. So the error term grows like $O(2^{3/2 m} \delta)$. Suppose $\delta < 1/2^{5/2m}$ then using the same line of reasoning as in the proof of \ref{thm:perm-det} we obtain the gap within precision $1/2^{\Omega(n)}$ which is already a gap\textbf{P} hard problem.
\end{proof}

This corollary may not seem very surprising, but it is not trivial since for any $\delta > 0$ (even for $\delta = 1/ 2^{2^n}$!) there are quantum states that have very high exact rank and very small approximate rank with parameter $\delta$. To see this, let $\delta$ be fixed according to the above and let $\ket \psi$ be a quantum state with $\chi_0(\psi) = 2^{\Omega (n)}$ that is orthogonal to $\ket 0$; then $\ket \phi = \delta \ket \psi + \sqrt {1-\delta^2} \ket 0$ has exact rank $2^{\Omega(n)}$ and approximate rank $1$ with parameter $O(\delta) > 0$.

Can we obtain the same lower bound for when $\delta$ is a constant? The result of \cite{Bravyi_Gosset_2016} can be re-stated as the following.

\begin{theorem}
If $\chi_\delta(\ket T^{\otimes n})$ is upper bounded by a polynomial, then there exists a \textbf{BPP}$/$poly algorithm to sample from the output distribution of arbitrary polynomial-size quantum circuits within $O(\delta)$ total variation distance. 
\label{thm:BG-TV}
\end{theorem}

The hardness of sampling within total variation distance has been the subject of extensive research over the past decade. In particular, Aaronson and Arkhipov \cite{AA10} demonstrated that assuming a conjecture about the anti-concentration and a conjecture about the hardness of approximating the permanent of Gaussian matrices, sampling within total variation distance from arbitrary quantum computations is hard unless the polynomial hierarchy collapses. We use a variant of this result based on hardness of sampling from random circuits \cite{arute2019quantum,AaronsonChen, bouland2017complexity} to deduce the following
\begin{theorem}[Approximate rank and polynomial hierarchy]
Assuming that for $0.1$ fraction of quantum circuits $U$ with polynomial-size approximating $\langle x |U| 0\rangle$ within $\delta$ multiplicative error is $\#$\textbf{P}-hard, stabilizer rank $\chi_\delta(\ket T^{\otimes n})$ is super-polynomial in $n$ unless the polynomial hierarchy collapses to the fourth level.
\end{theorem}

\begin{proof}
    Suppose the stabilizer rank is bounded by a polynomial. We show that there exists a protocol within $\textbf{P}^{\textbf{NP}}/$poly which can approximate $\langle 0| U|0\rangle$ within multiplicative error $O(\delta)$ for a $\geq 0.1$ fraction of quantum circuits. If we assume this task is $\#$\textbf{P}-hard, then using Toda's result \cite{toda1991pp}, $\textbf{PH} \subseteq \textbf{P}^{\#\textbf{P}} \subseteq \textbf{P}^{\textbf{NP}}/$poly. This implies $\mathbf{\Sigma}_2 \subseteq \textbf{PH} \subseteq \textbf{P}^{\textbf{NP}}/\poly \subseteq \mathbf{\Pi}_2/\poly$ which implies the collapse of the polynomial hierarchy by the Corollary 5.7 of \cite{cai2005competing}: $\mathbf{\Sigma}_2 \subseteq \mathbf{\Pi}_2/poly \implies \textbf{PH} \subseteq \mathbf{S}_2^{\mathbf{\Sigma}_2} \subseteq \textbf{ZPP}^{\mathbf{\Sigma}_3} \subseteq \mathbf{\Sigma}_4$; here $\mathbf{S}_2$ is called the symmetric alternation complexity class (see the \href{https://complexityzoo.net/Complexity_Zoo:S#s2p}{complexity zoo} for definition).

     Now we choose a family of polynomial size random quantum circuits $\mu$ such that $\mu$ is an approximate $2$-design and for any string $x$ for $> 0.1$ of $U \sim \mu$ approximating the output probability $p_{x,U} := |\bra x | U \ket 0|^2$ is hard for $\#$\textbf{P}. We use the algorithm by Bravyi and Gosset \cite{Bravyi_Gosset_2016} we referred to in Theorem \ref{thm:BG-TV} which samples from a circuit within $O(\delta)$ total variation distance in polynomial time, assuming the stabilizer rank of $T$ states are bounded by a polynomial. Taking the description of the approximate stabilizer decomposition of $\ket T^{\otimes n}$, this algorithm can be implemented in $\textbf{BPP}/\poly$. Using standard results based on Stockmeyer's approximate counting \cite{stockmeyer1983complexity} and the Paley-Zygmund anti-concentration bound (see for example Theorem 3 of \cite{harrow2023approximate}) this implies that there is a $\textbf{FBPP}^{\textbf{NP}}/$poly $=\textbf{FP}^{\textbf{NP}}/$poly algorithm to compute $\bra 0 U \ket 0$ for $1/8 - o(1)$ fraction of circuits from $\mu$. In the previous line, we have used Adleman-Bennet-Gill's \cite{adleman1978two, charles1981relative} $\textbf{BPP} \subseteq \textbf{P}/$poly.
     

\end{proof}

There have been many quantum proposals for demonstrating the hardness of sampling from quantum computers within total variation distance assuming plausible conjectures \cite{bremner2011classical, bouland2017complexity}. Some of these proposals were implemented recently in scales larger than what was available previously \cite{arute2019quantum, zhong2020quantum}.

\begin{remark}
Define the multiplicative distance between two quantum states with $n$ qubits with respect to stabilizer projectors as $d(\ket A, \ket B)_ {mult} := \max_{G \trianglelefteq \mathcal P_n} \max\pr*{\frac{\bra A \Pi_G \ket A}{\bra B \Pi_G \ket B}, \frac{\bra B \Pi_G \ket B}{\bra A \Pi_G \ket A}} - 1$, where the maximization is over Abelian subgroups of Pauli and $\Pi_G$ is the projector onto the subspace stabilized by $G$. If we define closeness between two quantum states according to this multiplicative distance, then we can place a super-polynomial lower bound on the approximate rank with respect to this measure assuming the same assumption of Theorem \ref{thm:perm-det}, i.e., $\textbf{P}^{\# \textbf{P}} \subseteq \textbf{P}/poly \implies \textbf{P}^{\# \textbf{P}} = \textbf{MA}$. To see this, we recall the first algorithm in \cite{Bravyi_Gosset_2016} where the authors give a Monte-Carlo procedure for estimating the amplitudes of quantum circuits with $t$ number of $T$ gates and polynomially many Clifford gates within relative error. In their procedure, they write the amplitude of any such circuit as a number proportional to $\bra{T} ^{\otimes t}\Pi_G \ket{T}^{\otimes t}$ for some stabilizer subgroup $G$ (see Equation~(9) of \cite{Bravyi_Gosset_2016}). As a result, if the approximate rank of the state $\ket \psi$ within multiplicative error is polynomial, then we repeat the procedure in \cite{Bravyi_Gosset_2016} assuming $\ket \psi = \ket{T}^{\otimes t}$ to get an estimation to a $\#$\textbf{P}-hard problem such as the gap of a Boolean function within relative error. This implies $\textbf{P}^{\#\textbf{P}} \subseteq \textbf{BPP}/\poly \subseteq \textbf{P}/\poly$ and hence $\textbf{P}^{\#\textbf{P}} = \textbf{MA}$ using \cite{babai1991non}.
\end{remark}

\subsection{Related works}
\label{sec:prev}
As we mentioned before, there have been several results achieving lower bounds on the stabilizer rank prior to this work. Here, we briefly review some of these results. The first work in this line of research was \cite{Bravyi_Smith_Smolin_2016} where the authors proved a lower bound of $\Omega(\sqrt{n})$ on the exact stabilizer rank of magic state with proof techniques similar to ours. More precisely, they constructed an arbitrary $n$-qubit quantum state with exact stabilizer rank $\Omega(n)$, which can be prepared by $O(n^2)$ $T$ gates and Clifford gates. Then, they used the quantum teleportation idea to show that the exact stabilizer rank of $O(n^2)$ magic states is at least $\Omega(n)$. The authors of \cite{Bravyi_2019} established an exponential lower bound in a restricted setting in which we only consider stabilizer states that are tensor products of $\ket{0}$ and $\ket{+}$ states. They used ultra-metric matrices machinery to characterize the Gram matrix of such stabilizer states. 

In \cite{Peleg2022lowerbounds} the authors proved a linear lower bound on the exact stabilizer rank by carefully investigating coefficients in the computational basis of any linear combination of $o(n)$ stabilizer states. In particular, they showed that there are two vectors in computation basis $\ket{x}$ and $\ket{x'}$ such that the number of ones in $x$ and $x'$ is different, but their corresponding coefficients are the same, so it cannot be associated with appropriate magic states. As a part of the same study, they also proved a lower bound of $\tilde \Omega(\sqrt{n})$ on the approximate stabilizer rank of magic states using tools from the polynomial method and an analysis of Boolean functions. 

Labib in ~\cite{Labib2022stabilizerrank} proved a linear lower bound on exact rank using tools from higher-order Fourier analysis. He showed that $\ket{T}^{\x n}$ has an exponentially small correlation with a class of functions, so-called ``quadratic non-classical phases'' defined in higher-order Fourier analysis, and are inherently connected to stabilizer states. Then, he showed that any function written as $o(n)$ of such functions cannot have such a small correlation with all quadratic non-classical phases. Finally, among other contributions, the authors of \cite{Lovitz2022newtechniques} almost re-derived the results of \cite{Peleg2022lowerbounds} using a modified version of a result in number theory on the subset-sum representation of a sequence of numbers with exponentially increasing subsequence. However, their lower bound on exact rank was $O(n/\log n)$ instead of $O(n)$, they handled both exact and approximate using similar approaches.

\subsection{Discussions and open questions}
\begin{enumerate}

    \item \textbf{Strengthening the bounds:}
    In this work, we provided a lower bound of $\frac{\Omega(n^2)}{\text{poly}\log n}$ on the approximate stabilizer rank of $\ket T^{\otimes n}$ or any state in the second Clifford hierarchy. We suspect that with more careful analysis, one can remove the $\poly\log(m)$ factor from our lower bound. We may also be able to strengthen our result to hold for the approximate rank of all magic states; right now, the bound works only for the $\ket{T}$ state (and its Clifford equivalents) and/or exact rank. Nevertheless, obtaining a super-quadratic lower bound using our approach would be much more challenging. Indeed, with our proof technique, any deterministic or randomized construction of quantum states with low ``non-Clifford complexity'' (defined in~Definition~\ref{def:tau}) but high approximate stabilizer rank is of interest (See also Conjecture \ref{conj:stab-complx}). However, at least for two natural classes of probability distributions over quantum states, i.e., the Haar measure and $t$-designs, it seems that the non-Clifford complexity of a ``typical'' quantum state grows at least quadratically with approximate stabilizer rank.
    
    \item \textbf{Proving an exponential lower bound:} 
    In our approach, we realize that the instances corresponding to states with high stabilizer rank and low circuit complexity correspond to an exponentially small fraction of quantum states. As a result, we believe that one needs to probe the structure of the stabilizer states more closely in order to make progress on this result. Building on the work of Labib \cite{Labib2022stabilizerrank}, we suggest the following sufficient condition to improve the bound on stabilizer rank to exponential.

\begin{conjecture}
Let $\ket \psi$ be a quantum state with stabilizer rank $r$, such that $F(\ket \psi) := \max_{s \in \Stab_n} |\langle s| \psi\rangle|^2 < 1/e^{\Omega(n)}$, then there exists a stabilizer state $\ket s$ such that $|\langle s | T\rangle^{\x n}|^2 \geq \frac 1 {\poly(r)}$.
\end{conjecture} 

For instance, we know that $F(\ket T^n) = \cos(\pi/8)^{2n}$ \cite{Bravyi_2019}; however, the best bound we can prove for the largest overlap is $1/4^r$, which fails at giving a super-linear lower bound. Our intuition is that to improve this bound, we need to find sharp bounds on the geometry of stabilizer states, e.g., given a stabilizer state, find how many stabilizer states there are that have at least $1 - \epsilon$ fidelity with that stabilizer state. 

Another venue for going beyond quadratic lower bounds is by lower bounding the stabilizer rank of $\ket \psi \otimes \ket \phi$ when $\ket \psi$ and $\ket \phi$ are $n/2$ qubit states sampled from the Haar measure. In particular, if we can show that $\chi_\delta (\ket \psi \otimes \ket \phi) > (\chi_\delta (\ket \psi) \chi_\delta(\ket \phi))^{1/2 + c}$ for some constant $c > 0$ and $\delta = 1/2^{\poly(n)}$, 
then we can improve the quadratic lower bound on the stabilizer rank to possibly even exponential. The intuition is that in order to prepare $\ket \psi \otimes \ket \phi$ we quadratically have fewer number of $\ket T$ states than a $n$ qubit state sampled from the Haar measure. We note that we can show $\chi_\delta (\ket \psi \otimes \ket \phi) \approx \chi_\delta (\ket \psi)  \chi_\delta(\ket \phi)$ for $\delta = 1/2^{2^{\Omega(n)}}$ which is not sufficient for our purpose.

\item \textbf{Strong lower bounds on exact rank from weak lower bounds on approximate ranks:} Another interesting question is whether our lower bound on approximate stabilizer rank has any implications for exact stabilizer rank. Lemma~\ref{lem:gap} shows that in general, we cannot get a super-quadratic lower bound on exact stabilizer rank from only knowing that the approximate stabilizer rank is $\Omega(m^2 /\poly\log(m))$. However, we might be able to use some additional structures of magic states to improve the lower bound. 

\item \textbf{Complexity theoretic connections:} A natural open question is whether either of the ideas used in this paper can be utilized to prove the quadratic uncertainty principle for the AND function; see Section \ref{sec:quad-uncertainty}. Our main theorem implies a lower bound on the decomposition of the Boolean function $M(x) = 1$ iff $|x| = 0 \mod 8$. Rather surprisingly, this result can be viewed as a lower bound on a classical problem from an upper bound on a quantum state synthesis problem. Understanding the full capabilities of this approach for classical problems is an interesting future direction.

As we discussed in Section \ref{sec:cond} since the amplitudes of quantum circuits encode arbitrary gap\textbf{P} problems, a bound on the exact rank provides a way of decomposing an arbitrary gap\textbf{P} function into a sum of polynomially many computable functions. We know that sampling one bit from a Clifford computation can be done using a $\oplus$\textbf{L} machine. If we can show that a strong simulation of Clifford quantum circuits is doable in a complexity class like gap\textbf{L}, then a polynomial upper bound on the stabilizer rank can be translated as the number of determinant computations needed in order to compute the permanent. The same applies to approximate rank; the approximation parameter is translated into an additive error in the above decomposition. In the field of enumerative combinatorics, a decomposition of the permanent into a summation over efficiently computable structures is called a permanent identity; many such identities are known e.g., Glynn, Ryser, MacMahon (See \cite{chabaud2022quantum}).


For the approximate rank, as we discussed in Section~\ref{sec:cond}, a lower bound of $r$ on the approximate rank of a magic state implies sampling within total variation distance $O(\delta)$ from quantum circuits in time proportional to $\poly(n) \cdot r$, and assuming conjectures about hardness of approximate counting for specific functions implies a super-polynomial lower bound on the approximate rank assuming the collapse of polynomial hierarchy \cite{AA10, bouland2017complexity, bremner2011classical}. Can we improve this result by basing the lower bound merely on the non-collapse of the polynomial hierarchy? We remark that unlike results such as \cite{AA10, AaronsonChen, bouland2017complexity, bremner2011classical} we can allow error correction and distributions that are not necessarily from random constrained circuits. 
\end{enumerate}

\section{Preliminaries}
Let $\F$ be the finite field of order two. 
$\F^n$ is a vector space over $\F$. 
An affine subspace of $\F^n$ is a linear subspace shifted by an arbitrary vector. 
A quadratic function over $\F^n$ is of the form $x\mapsto x^TAx + a^Tx$ where $T$ is an $n\times n$ matrix and $a\in \F^n$. 
The Hilbert space corresponding to an $n$-qubit system is $(\C^2)^{\otimes n}$. 
We identify the computational basis for this Hilbert space by elements of $x\in\F^n$. Let $I_2$ be the identity function on $\C^2$. Let $\mathcal{H}$ be a finite-dimensional Hilbert space in the following definitions. 
Let $U(\mathcal{H})$ denote all unitaries acting on $\mathcal{H}$. 
Let $\Proj{\mathcal{H}, M}$ denote all the orthogonal projections acting on $\mathcal{H}$ and rank at most $M$. For a linear operator $A:\mathcal{H} \to \mathcal{H}$, $\norm{A}$ denotes the operator norm of $A$ defined as
\begin{align}
    \|A\| := \sup_{\ket{\phi}\in\mathcal{H}\setminus \set{0}} \frac{\norm{A\ket{\phi}}}{\norm{\ket{\phi}}}.
\end{align}
Is the largest singular value of $A$ and is the same as the infinite Schatten norm. Here, by $\|\ket v\|$ we mean the $2$-norm of the vector $\ket v$. We also denote the trace-norm of $A$ by $\norm{A}_1 = \tr{\sqrt{A A^\dagger}}$.
A \emph{quantum channel} is a linear super-operator that is trace-preserving and completely positive. The diamond distance between two quantum channels $\Phi$ and $\Psi$ for a Hilbert space $\mathcal{H}$ is defined as
\begin{align}
    \sup_{d\geq 0}\sup_{ X\in \mathcal{H}\x \C^d} \frac{\norm{(\Phi\x\textnormal{id}_d(X)-\Psi\x\textnormal{id}_d(X))}_1}{\norm{X}_1},
\end{align}
where $\textnormal{id}_d$ is the identity channel on $\C^d$.

\subsection{Quantum circuits}
A quantum circuit on $n$ qubits is a unitary operation over $(\C^2)^{\otimes n}$. 
In most cases, we represent a quantum circuit as a product of $m$ unitaries $V_1, \cdots, V_m$ where each $V_i$ is a unitary acting on a subset of qubits (typically of size one, two, or three). 
The unitaries acting on a smaller number of qubits are called quantum gates, and the set of all allowed gates is called a gate set. 
A gate set is \emph{universal} if one can approximate any unitary on $n$ qubits within an arbitrary error using a sequence of gates from the gate set. We know that the gate set $\set{H, CNOT, S, T}$ is universal where
\begin{align}
    H &= \frac{1}{\sqrt{2}} \begin{bmatrix}1&1\\1&-1
    \end{bmatrix}\\
    S &=\begin{bmatrix}1&0\\0&i
    \end{bmatrix}\\
    CNOT &=\begin{bmatrix}1&0&0&0\\0&1&0&0\\ 0&0&0&1\\ 0&0&1&0\\
    \end{bmatrix}\\
    T &= \begin{bmatrix}
        1&0\\0&{e^{i\pi/4}}
    \end{bmatrix}.
\end{align}

Finally, we define a notion of ``non-Clifford complexity,'' quantifying how many non-Clifford resources we need to prepare a quantum state. 
\begin{definition}
\label{def:tau}
Let $\ket \phi$ be an $n$ qubit state. $\tau_\epsilon(\ket{\phi})$ is the smallest $k$ such that there exists a quantum circuit $V$, acting on $n + \lambda$ qubits for $\lambda >0$, consisting of arbitrary number of Clifford gates and $k$ number of  $T$ gates such that $\norm{\ket{\phi} \ket {0^{\lambda}} - V(\ket{0^n} \ket{0^{\lambda}}) }\leq \epsilon$.
\end{definition}

\subsection{Stabilizer formalism}
\label{sec:stab-form}
We briefly review stabilizer formalism here (see \cite{Nielsen_2012} for details). Let $\mathcal{P}_n$ be the Pauli group acting on $n$ qubits. The Clifford group on $n$ qubits, denoted by $\mathcal{C}_n$, is the normalizer of $\mathcal{P}_n$ in the unitary group modulo a phase. We denote by $\Stab_n$  the set of all stabilizer states, which can be written of the form \cite{VanDenNest_2010}
\begin{align}
    \frac{1}{\sqrt{|A|} }\sum_{x\in A} i^{\ell(x)}(-1)^{Q(x)} \ket{x},
\end{align}
where $A\subset \F^n$ is  an affine subspace of $\F^n$ ($A = \{L y + v: y \in \F^m\}$, where $L \in \F^{n \times m}$, $V \in \F^n$), $\ell:\F^n \to \F$ is a linear function, and $Q:\F^n\to \F$ is a quadratic function. Here are some known facts about stabilizer states.
\begin{lemma}
\label{lm:stabilizer}
The following are true
    \begin{enumerate}
    \item For any Clifford unitrary $C\in\mathcal{C}_n$ and any stabilizer state $\ket{s}\in\Stab_n$, we have $C\ket{s}\in\Stab_n$.
    \item Let $\ket{s}$ be an $n$-qubit stabilizer state and $b\in\F$. Then,  $I_2^{\otimes n-1} \otimes \bra{b} \ket{s}$ is either zero or proportional to a state in $\Stab_{n-1}$.
    \item $\abs{\Stab_n}\leq e^{0.54 n^2}$ for $n\geq 6$ \cite{WinterMagic2022}.
\end{enumerate}
\end{lemma}

We next define the (approximate) stabilizer rank of a quantum state. Let $\ket{\phi}$ be an $n$-qubit state. We denote by $\chi(\ket{\phi})$ the stabilizer rank of $\ket{\phi}$ defined as the minimum $M>0$ such that there exists $c_1, \cdots, c_M \in \C$ and $\ket{s_1}, \cdots, \ket{s_M}\in \Stab_n$ such that $\ket{\phi} = \sum_{i=1}^M c_i \ket{s_i}$. We also define $\chi_\delta(\ket{\phi})$ as 
\begin{align}
    \min_{\ket{\psi}: \norm{\ket{\phi}-\ket{\psi}}\leq \delta} \chi(\ket{\psi}),
\end{align}
where $\ket{\psi}$ does not need to be normalized. We remark that the stabilizer rank is operationally relevant to the cost of classical simulation. We may choose other measures of closeness to the set of stabilizer states which we will review in Appendix \ref{a:measures}.

\subsection{Haar measure and $t$-Designs}
Let $\mathcal{H}$ be a finite-dimensional Hilbert space. The Haar measure over $\mathcal{H}$ is the unique probability distribution on the unit vectors in $\mathcal{H}$ invariant under the action of any unitary in $U(\mathcal{H})$. We shall use the following concentration of measure result for the Haar measure over finite-dimensional Hilbert spaces known as L\'evy's theorem \cite[Theorem 7.37]{watrous2018theory}.
\begin{theorem} [L\'evy's concentration]
    \label{th-Levy}
    Let $\mathcal{H}$ be a $d$-dimensional Hilbert space and $\ket{\phi}$ be distributed according to the Haar measure in $\mathcal{H}$. For any $\kappa$-Lipschitz function $f$ from the unit sphere in $\mathcal{H}$ to $\R$ and $\epsilon > 0$, we have 
    \begin{align}
        \P{f(\ket{\phi}) - \E{f(\ket{\phi})} \geq \epsilon} \leq 2e^{-\frac{\epsilon^2d}{25\pi\kappa^2}}
    \end{align}
\end{theorem}

A function is $\kappa$-Lipschitz if $|f(\ket \phi) - f(\ket \psi)| \leq \kappa \|\ket \psi-\ket \phi\|$. We use the following property of Haar measure, whose proof is similar to that of \cite[Lemma 7.2]{watrous2018theory}.
\begin{lemma}
\label{lm:haar-expectation}
    Let $\mathcal{H}$ be a $d$-dimensional Hilbert space and $P$ be a projection on $\mathcal{H}$ of rank $M$. Let $\ket{\phi}$ be distributed according to the Haar measure in $\mathcal{H}$. Then, 
    \begin{align}
        \tr\pr{P^{\x t} \E{(\ketbra{\phi}{\phi})^{\x t}}} = \frac{\binom{M + t - 1}{t}}{\binom{d + t - 1}{t}} = \frac{(M+t -1) \cdots (M+1)M}{(d+t-1)\cdots (d+1)d}.
    \end{align}
\end{lemma}

The Haar measure over $U(\mathcal{H})$ is the unique probability distribution invariant under left or right multiplication by any unitary in $U(\mathcal{H})$.  While preparing a unitary according to the Haar measure requires an exponential amount of resources \cite{knill1995approximation}, unitary $t$-designs, which we formally define here, mimic the Haar measure up to the $t$-th moment and can be efficiently prepared for $t$ small enough. 
\begin{definition}
\label{def:t-design}
    Denote for  a distribution $\nu$ over $U(\mathcal{H})$
    \begin{align}
        M_t^{\nu}(\rho) \eqdef \int U^{\otimes t} \rho (U^\dagger)^{\otimes t} d\nu,
    \end{align}
    which is a quantum channel. The distribution $\nu$ is called an $\epsilon$-approximate $t$-design if 
    \begin{align}
        \norm{M_t^{\nu} - M_t^{\textnormal{Haar}}}_{\diamond} \leq \epsilon.
    \end{align}
\end{definition}
There are several constructions  of approximate $t$-designs \cite{Harrow_Low_2009, Nakata_Hirche_Koashi_Winter_2017, Haferkamp2022randomquantum, Roth_2023, Harrow_Horodecki_2016}. We state here two constructions: one that has a close connection to stabilizer formalism and one using random circuits.
\begin{theorem}[\cite{Roth_2023}]
\label{th:t-design-exsitence}
    There exists constant $C_1$ and $C_2$ such that for all $\epsilon$, $n$, and $t$ with $n \geq C_2 t^2$ the following holds. There exists an $\epsilon$-approximate $t$-design $\nu$ such that any unitary in the support of $\nu$ consists of Clifford gates and at most $C_1 \log^2(t)(t^4 + t\log(1/\epsilon))$ number of $T$ gates. 
\end{theorem}

\begin{theorem} [\cite{Haferkamp2022randomquantum}]
\label{th:t-design-exsitence2}
    For $n\geq 2\log(4t) + 1.5\sqrt{\log(4t)}$, there exists an $\epsilon$-approximate $t$-design for $n$ qubits such that  each unitary in the support of $\nu$ is composed of 
    \begin{align}
        C n\ln^5(t) t^{4+3\frac{1}{\sqrt{\log(t)}}}(2nt + \log(1/\epsilon)),
    \end{align}
    two-qubit gates for absolute constant $C>0$.
\end{theorem}

We also state a conjecture on the optimal number of non-Clifford gates in any $t$-design when $t$ scales with $n$. The intuition behind this conjecture is that by \cite[Proposition~8]{Harrow_Horodecki_2016}, we need $\tilde \Omega(nt)$ gates to get a $t$-design for $n$-qubits, and we expect that most of these gates should be non-Clifford.
\begin{conjecture}
\label{con:converse-t-design}
Let $\nu$ be a distribution supported on quantum circuits with an arbitrary number of Clifford gates and $k$ number of $T$ gates. If $\nu$ is an $\epsilon$-approximate $t$-design for $t=\omega(1)$, then $k = \Omega(t)$.
\end{conjecture}

\section{Lower bounds on approximate stabilizer rank of magic states }
In this section, we state and prove our main result, which is a lower bound on the approximate stabilizer rank of $\ket{T}^{\x n}$.
\begin{theorem}
\label{th:lower-bound-stab}
Let $1>\delta>0$. We have 
\begin{align}
    \chi_\delta(\ket{T}^{\otimes m}) = \Omega\pr*{\frac{(1-\delta^2)^2m^2}{\poly\log (m)}}.
\end{align}
\end{theorem}

We prove the above result using three ingredients.
\begin{enumerate}
    \item In Section~\ref{sec:ing1}, we show  that for each $n$ and $\delta$ there exists an arbitrary $n$-qubit quantum state with an approximate stabilizer rank $\Omega(2^n/n^2)$ (Lemma~\ref{lm:random-stab-rank}). To prove this result, we consider the approximate stabilizer rank of a random $n$-qubit state distributed according to the Haar measure. For each choice of $\Omega(2^n/n^2)$ number of stabilizer states, we obtain a doubly exponential upper bound on the probability that we can estimate the random state as a linear combination of those stabilizer states. We then use the union bound to obtain an upper bound on the probability that a random state has approximate stabilizer rank $\Omega(2^n/n^2)$.

    \item In Section~\ref{sec:ing2}, we state a result of \cite{low2018trading} that shows every $n$-qubit state can be approximated using Clifford gates and at most $O(\poly(n 2^{n/2})$ number of  $T$ gates and many ancilla qubits (Lemma~\ref{lm:t-complex2}).
    
    \item In Section~\ref{sec:ing3}, we finally use the ideas in gadget-based implementation of $T$ gates to show that the approximate stabilizer rank of the output of state is upper bounded by stabilizer rank of $\ket{T}^{\otimes m}$ (Lemma~\ref{lm:stab-rank-circuit}).
\end{enumerate}
Having these ingredients, we prove Theorem~\ref{th:lower-bound-stab} in Section~\ref{sec:proof-lb}, but provide the high-level argument here. Fixing $m$, we choose $n$ such that $\poly(n) 2^{n/2} \approx m$, which implies that $n \approx 2 \log m$  and $2^{n/2} \approx m/\poly \log (m)$. We then choose quantum state $\ket{\phi}$ with $n$ qubits and approximate stabilizer rank $\Omega((1-\delta^2)^22^n/n^2) = \Omega((1-\delta^2)^2m/\poly\log (m)) $ according to the first ingredient. Let $V$ be the quantum circuit according to the second ingredient, i.e., $V\ket{0^n} \approx \ket{\phi}$ and $V$ contains Clifford gates and $m = \poly(n)2^{n/2}$ number of $T$ gates. Then, using the last ingredient, we obtain
\begin{align}
    \Omega((1-\delta^2)^2m^2/\poly\log(m)) = \chi_\delta(\ket{\phi}) \approx \chi_{\delta}(V(\ket{0^n}  |0^\lambda\rangle)) \leq \chi_\delta(\ket{T}^{\x \poly(n) 2^{n/2}}) \approx  \chi_\delta(\ket{T}^{\x m}).
\end{align}
See Section~\ref{sec:proof-lb} for details.

\subsection{Existence of quantum states with large approximate stabilizer rank} 
\label{sec:ing1}
The following lemma provides an upper bound on the likelihood that a Haar random state of $n$ qubits has a small approximate stabilizer rank.
\begin{lemma}
    \label{lm:random-stab-rank}
    Let $\ket{\phi}$ be a random $n$-qubit state distributed according to Haar measure with $n\geq 6$. Let $M$ be a positive integer and $0<\delta<1$ be such that $1-\delta^2 - \frac{M}{2^n} > 0$.
    We have 
    \begin{align}
        \P{\chi_\delta(\ket{\phi}) \leq M} \leq 2e^{0.54n^2M - \frac{(1-\delta^2 - M/2^n)^22^n}{100\pi}}
    \end{align}
    In particular, for  $n \geq 2\log \frac{1}{1-\delta^2} + 9$, there exists an $n$-qubit state $\ket{\phi}$ with
    \begin{align}
        \chi_\delta(\ket{\phi}) \geq  C \frac{(1-\delta^2)^22^n}{n^2}
    \end{align}
    for an absolute constant $C \geq \frac{1}{1000}$.
\end{lemma}
\begin{remark}
    A random state distributed according to Haar measure with probability zero lies inside the span of any $2^n - 1$ stabilizer states. Therefore, by union bound,
    \begin{align}
        \P{\chi(\ket{\phi}) = 2^n} = 1.
    \end{align}
    Lemma~\ref{lm:random-stab-rank} states a robust version of this observation.
\end{remark}
To prove the above lemma, we first introduce a necessary condition for the approximate stabilizer rank of an arbitrary state to be less than $M$ in the following lemma. 

\begin{lemma}
    \label{lm:stab-rank-nec}
    Let $\ket{\phi}$ be an $n$-qubit state.  If $\chi_\delta(\phi) \leq M$, then $\max_{S\subset\Stab_n: |S| = M} \norm{P_S\ket{\phi}}^2 \geq 1 - \delta^2$ where $P_S$ denotes the orthogonal projection onto the subspace spanned by the elements of $S$.
\end{lemma}
\begin{proof}
    Assume that $\chi_\delta(\phi) \leq M$. It means that there exist $c_1, \cdots, c_M \in \C$ and $S = \set{\ket{s_1}, \cdots, \ket{s_M}} \subset \Stab_n$ such that 
    \begin{align}
        \norm*{\ket{\phi} - \sum_{i=1}^M c_i \ket{s_i}} \leq \delta
    \end{align}
    We also know that 
    \begin{align}
        1 
        &= \norm{\ket{\phi}}^2\\
        &\stackrel{(a)}{=} \norm{\ket{\phi} - P_S \ket{\phi}}^2 + \norm{P_S \ket{\phi}}^2\\
        &\stackrel{(b)}{\leq}   \norm*{\ket{\phi} - \sum_{i=1}^M c_i \ket{s_i}}^2+ \norm{P_S \ket{\phi}}^2 \\
        &\leq \delta^2 + \norm{P_S \ket{\phi}}^2
    \end{align}
    where $(a)$ follows from Pythagoras Theorem in Hilbert Spaces, and $(b)$ follows since $P_S\ket{\phi}$ is the closest point to $\ket{\phi}$ in the subspace spanned by elements of $S$. Therefore, $\norm{P_S \ket{\phi}}^2 \geq 1-\delta^2$ as desired.

\end{proof}
\begin{proof}[Proof of Lemma~\ref{lm:random-stab-rank}]
    Before delving into the technical proof, we sketch the main ideas. Lemma~\ref{lm:stab-rank-nec} helps us to upper bound the probability that the approximate stabilizer rank of a random state is less than $M$ in terms of the norm of the projection of the random state into several subspaces. We use union bound and L\'evy's theorem (Theorem \ref{th-Levy}) to upper bound that probability.
    
    Note that
    \begin{align}
        \P{\chi_\delta(\ket{\phi}) \leq M} 
        &\stackrel{(a)}\leq \P{\max_{S\subset\Stab_n: |S| = M} \norm{P_S\ket{\phi}}^2 \geq 1 - \delta^2}\\
        &\stackrel{(b)}\leq \sum_{S\subset\Stab_n: |S| = M} \P{\norm{P_S\ket{\phi}}^2 \geq 1 - \delta^2}\\
        &\leq \binom{|\Stab_n|}{M} \sup_{P\in \Proj{(\C^2)^{\otimes n}, M}} \P{\norm{P\ket{\phi}}^2 \geq 1 - \delta^2}\\
        &\leq|\Stab_n| ^M \sup_{P\in \Proj{(\C^2)^{\otimes n}, M}} \P{\norm{P\ket{\phi}}^2 \geq 1 - \delta^2}\\
        &\stackrel{(c)}{\leq} e^{0.54n^2M} \sup_{P\in \Proj{(\C^2)^{\otimes n}, M}} \P{\norm{P\ket{\phi}}^2 \geq 1 - \delta^2}\label{eq:union-bound}
    \end{align}
    where $(a)$ follows from Lemma~\ref{lm:stab-rank-nec}, $(b)$ follows from the union bound, and $(c)$ follows from Lemma~\ref{lm:stabilizer}.
We now upper bound $\P{\norm{P\ket{\phi}}^2 \geq 1 - \delta^2}$ for $P\in\Proj{(\C^2)^{\otimes n}, M}$ using L\'evy's theorem (Theorem~\ref{th-Levy}).

We take $\mathcal{H} = (\C^2)^{\otimes n}$ and $f(\ket{\phi}) = \norm{P\ket{\phi}}^2$ and show that $f$ is $2$-Lipschitz and $\E{f(\ket{\phi})} \leq \frac{M}{2^n}$. First note that for arbitrary unit vectors $\ket{\phi}$ and $\ket{\psi}$, we have
\begin{align}
    \abs{f(\ket{\phi}) - f(\ket{\psi})} 
    &= \abs{\norm{P\ket{\phi}}^2 - \norm{P\ket{\psi}}^2}\\
    &= \abs{\norm{P\ket{\phi}} - \norm{P\ket{\psi}}}(\norm{P\ket{\phi}} + \norm{P\ket{\psi}})\\
    &\leq \norm{P\ket{\phi}-P\ket{\psi}}(\norm{P\ket{\phi}} + \norm{P\ket{\psi}})\\
    & \leq \norm{\ket{\phi}-\ket{\psi}}(\norm{\ket{\phi}} + \norm{\ket{\psi}})\\
    &= 2\norm{\ket{\phi}-\ket{\psi}},
\end{align}
which means that $f$ is 2-Lipschitz. By Lemma~\ref{lm:haar-expectation}, we also have
\begin{align}
    \E{f(\ket{\phi})} = \frac{\rank P}{2^n} \leq \frac{M}{2^n}.
\end{align}
Applying L\`evy's theorem, we obtain
\begin{align}
    \P{\norm{P\ket{\phi}}^2 \geq 1-\delta^2} \leq 2e^{-\frac{(1-\delta^2 - M/2^n)^22^n}{100\pi}}.\label{eq:levy-bound}
\end{align}
Combining \eqref{eq:union-bound} and \eqref{eq:levy-bound} yields the first part of Lemma~\ref{lm:random-stab-rank}. 

To prove the second part, it is enough to show that for $n\geq 2 \log \frac{1}{1-\delta^2} + 9$ and $M = \frac{(1-\delta^2)^22^n}{1000n^2}$,
\begin{align}
    \P{\chi_\delta(\ket{\phi}) \leq M } < 1,
\end{align}
 By the first part of the lemma, it is enough to show that
\begin{align}
    0.54n^2M - \frac{(1-\delta^2 - M/2^n)^22^n}{100\pi} \leq -1.
\end{align}
We have by direct calculation
\begin{align}
    1 - \delta^2 - \frac{M}{2^n}
    &= 1 - \delta^2 - \frac{(1-\delta^2)^2}{1000n^2}\\
    &\geq (1-\delta^2)\pr*{1-\frac{1}{1000}}.
\end{align}
and therefore
\begin{align}
    0.54n^2M - \frac{(1-\delta^2 - M/2^n)^22^n}{100\pi}
    &\leq 0.54n^2M - \frac{(1-\delta^2)^2 (0.999)^2 2^n}{100\pi}\\
    &= (1-\delta^2)^2 2^n\pr{5.4 \times 10^{-4} - \frac{0.999^2}{100\pi}}\\
    &\leq - 0.0026 (1-\delta^2)^2 2^n,
\end{align}
which is less than $-1$ for $n \geq 2\log \frac{1}{1-\delta^2} + 9$.
\end{proof}

\subsection{Trading $T$ gates for Clifford operations}
\label{sec:ing2}

In the previous section, we showed that, except for an exponentially small fraction, Haar random quantum states have exponentially large approximate ranks. Our strategy is to translate this lower into a lower bound on the approximate stabilizer rank of $\ket T^{\otimes n}$. As an intermediate step toward this goal, we need to find an upper bound on the number of $T$ gates necessary to sample from the Haar measure. This section finds an upper bound on the number of $T$ gates to construct arbitrary quantum states.

One might expect that an arbitrary quantum state would require $\Omega(2^n)$ numbers of $T$ gates to prepare. It turns out that by allowing ancilla qubits, arbitrary quantum states may be prepared using $\tilde O (2^{n/2})$ number of $T$ gates and $\tilde O (2^{n/2})$ ancillae.

\begin{lemma} [Trading $T$ gates with Clifford operations \cite{low2018trading}]
    Let $\ket \phi$ be an arbitrary $n$-qubit quantum state then $\tau_{4^{-n}} (\ket \phi) = 2^{n/2} O(n)$ where $\tau_\epsilon$ is defined in Definition \ref{def:tau}.
\label{lm:t-complex2}
\end{lemma}

\begin{proof} [Proof sketch of \cite{low2018trading}.] For completeness and to highlight certain construction features, we provide a brief sketch of the proof of this theorem; for details, see \cite{low2018trading}. We show that there exists unitary $U$ which is a product of $ 2^{n/2} O(n + \log \frac 1 \epsilon)$ number of $T$ gates and arbitrary number of Clifford operations on $n + \lambda$ qubits for $\lambda = \tilde O(2^{n/2})$, such that
    $$
    \|U \ket{0^n} |0^{\lambda}\rangle - \ket \phi |0^{\lambda}\rangle\| \leq \epsilon.
    $$  
The proof is in two steps. The first step gives a $T$ efficient way of constructing a certain data look-up oracle that queries classical data. In the second step, the data look-up oracle is used to query relevant rotation angles to rotate quantum bits one-by-one from qubit $i = 1$ to qubit $i = n$ to their proper positions; each round affects one qubit, and the looked-up data can be uncomputed exactly. We highlight that the first step is the main bottleneck in the number of $T$ gates. 

The data query oracle uses $\tilde O (2^{n/2})$ number of $T$ gates and use $\tilde O(2^{n/2})$ ancillae. This oracle takes an $n$-bit string $x$ as input and outputs a $b$ bit output string $a(x)$ according to $\ket {0^n} \ket {0^b} \ket {0^\lambda} \rightarrow \ket {0^n} \ket {a(x)} \ket {g_x}$, where $g_x$ can be viewed as a $\lambda$-qubit ``garbage'' state. First, assume a construction that uses $\lambda = 1$ ancilla. The oracle uses an operator SELECT$(x) = \ket x \bra x \otimes I + \ket x \bra x \otimes X$ which flips the ancilla bit only if the input data is provided as $x$, then conditioned on the ancilla having been flipped the circuit applies controlled $X$ gates to write the output $a(x)$ value on the second register; we call this second procedure WRITE$(x)$. The oracle does this for all $2^n$ input strings, and hence the $\lambda =1 $ construction uses $\tilde O(2^n)$ number of $T$ gates. For any $x$, SELECT$(x)$ uses $O(n)$ number of $T$ gates, but the WRITE$(a(x))$ operations are Clifford and hence are free. The main idea is to trade SELECT operations for WRITE operations by applying SELECT on the first $n/2$ bits of the use $\lambda = \tilde O (2^{n/2})$ ancilla bits to store all the possible $O(2^{n/2})$ outputs corresponding to the second half of the string. We then use $O(2^{n/2})$ controlled SWAP gates, each taking $O(n)$ number of $T$ gates, to move the right output bits in place of the $a$ register. The whole procedure uses $\tilde O(2^{n/2})$ number of $T$ gates.
\end{proof}

We remark that the above result makes a nontrivial use of ancilla qubits in order to obtain a quadratic improvement on the $T$ count for arbitrary quantum states. Without allowing ancillae, only weaker $\tilde O(2^n)$ upper bounds are known; see, for example, \cite[Theorem 3.3]{knill1995approximation}. Our overall strategy allows the use of ancilla qubits, so long as they start from stabilizer states and end with stabilizer states. We hence make a nontrivial use of Lemma~\ref{lm:t-complex2}. This lemma enables us to cross the linear lower-bound barrier on the stabilizer rank and obtain a quadratic lower bound.

\subsection{Approximate stabilizer rank and quantum circuits}
\label{sec:ing3}
\begin{lemma}
\label{lm:stab-rank-circuit}
    Let $V$ be a quantum circuit consisting of Clifford gates and $k$ $T$ gates. Then, 
    \begin{align}
        \chi_\delta(V\ket{0}) \leq \chi_\delta(\ket{T}^{\otimes k}).
    \end{align}
\end{lemma}
The proof is based on the idea of teleportation of each $T$ gate using a magic $\ket{T}$ state. We need the following two technical results for the measurement required in the teleportation of $T$ gate. We first show that under certain conditions, the measurement output has uniform distribution.\footnote{This result is implicitly assumed in the literature, but we provide proof for completeness.}
\begin{lemma}
    \label{lm:unif-T}
    Let $\ket{\psi}$ be an $n$-qubit quantum state. Apply controlled-not gate on the last qubit of $\ket{\psi} \otimes \ket{T}$ controlled on one of the qubits in $\ket{\psi}$ and then measure the last qubit in the computational basis. Then, the output of the measurement has a uniform distribution.
\end{lemma}
\begin{proof}
    Write $\ket{\psi} = \sum_{x\in\F^n} \alpha_x \ket{x}$. Without loss of generality, we assume that the CNOT is controlled on the first qubit. The state after applying CNOT is
    \begin{align}
        \sum_{x\in\F^n} \alpha_x \ket{x} \otimes \frac{\ket{x_1} + e^{i\pi/4}\ket{1-x_1}}{\sqrt{2}}
    \end{align}
    Hence, the probability that we obtain zero after measuring the last qubit is 
    \begin{align}
        \norm*{(I_2^{\otimes n} \otimes \bra{0})\sum_{x\in\F^n} \alpha_x \ket{x} \otimes \frac{\ket{x_1} + e^{i\pi/4}\ket{1-x_1}}{\sqrt{2}}}^2
        &= \frac{1}{2} \norm*{\sum_{x\in\F^n} e^{x_1i\pi/4}\alpha_x \ket{x} }^2\\
        &=\frac{1}{2} \norm{\psi}^2 = \frac{1}{2},
    \end{align}
    as desired.
\end{proof}
The next lemma states that performing a ``balance'' measurement on one of the qubits does not increase the approximate stabilizer rank for at least one outcome.
\begin{lemma}
    \label{lm:T-gadget}
     Consider measuring the last qubit in an $n$-qubit state $\ket{\phi}$ in the computational basis. Let $p_b$ be the probability that the measurement output is $b$ and $\ket{\widetilde{\phi}_b}$ be the post-measurement state when the output is $b$.
     If $p_0 = p_1$, then $\min_{b=0, 1} \chi_\delta(\ket{\widetilde{\phi}_b}) \leq \chi_\delta(\ket{\phi})$.
\end{lemma}

\begin{proof}
    Let $\chi_\delta(\ket{\phi}) = M$ and consider $c_1, \cdots, c_M \in \C$ and $\ket{s_1}, \cdots, \ket{s_M}\in \Stab_n$ such that
    \begin{align}
        \norm*{\ket{\phi} - \sum_{i=1}^M c_i \ket{s_i}} \leq \delta.
    \end{align} 
    We also define $\ket{\psi} = \sum_{i=1}^M c_i \ket{s_i}$ and decompose 
     $\ket{\phi} = \ket{\psi_0} \otimes \ket{0} + \ket{\phi_1} \otimes \ket{1}$ and $\ket{\psi} = \ket{\psi_0} \otimes \ket{0} + \ket{\psi_1} \otimes \ket{1}$.
     Since $\ket{0}$ and $\ket{1}$ are orthogonal, we have
     \begin{align}
         \norm{\ket{\phi_0} - \ket{\psi_0}}^2 +          \norm{\ket{\phi_1} - \ket{\psi_1}}^2 \leq \delta^2.
     \end{align}
     Hence,
     \begin{align}
         \min_{b=0, 1} \norm{\ket{\phi_b} - \ket{\psi_b}} \leq \frac{\delta}{\sqrt{2}}.
     \end{align}
     Without loss of generality, assume that the minimum is achieved for $b=0$. Since the post-measurement state on the first $n - 1$ qubits is $\sqrt{2} \ket{\phi_0}$ and $\norm{\ket{\widetilde{\phi}_0} - \sqrt{2}\ket{\psi_0}} \leq \delta$. Finally, we  note that 
     \begin{align}
         \sqrt{2}\ket{\psi_0} &= (I_2^{\otimes n - 1}\otimes \bra{0})\ket{\psi} \\
         &= \sum_{i=1}^M c_i (I_2^{\otimes n - 1} \otimes \bra{0}) \ket{s_i}.
     \end{align}
     which is a linear combination of at most $M$ stabilizer states. This proves that $\chi_\delta(\ket{\widetilde{\phi}_0}) \leq M$ as desired.
\end{proof}

\begin{proof}[Proof of Lemma~\ref{lm:stab-rank-circuit}]
    \begin{figure}
        \centering
        \begin{quantikz}
        & \qw\gategroup[2,steps=3,style={dashed,
        rounded corners,fill=blue!20, inner xsep=2pt},
        background]{}  & \ctrl{1}&  \gate{S} & \qw \\
        &\ket{T} & \targ{}  & \meter{} \vcw{-1}
        \end{quantikz}
        \caption{Gadget for implementing $T$ gate}
        \label{fig:gadget-T}
    \end{figure}
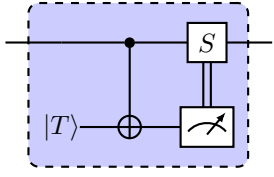
    We replace each $T$ gate in the circuit with the gadget introduced in \cite{Gottesman_Chuang_1999} as in Fig.~\ref{fig:gadget-T}. We index $n + k - i +1$ the second qubit in the gadget corresponding to the $i^{th}$ $T$ gate in $V$ and denote by $x_i\in\set{0, 1}$ the output of the measurement in that gadget. We note that $x = (x_1, \cdots, x_k)$ has uniform distribution over $\set{0, 1}^k$ by Lemma~\ref{lm:unif-T}. We can  formally express the equivalence of gadget-based implementation of as 
    \begin{align}
    V\ket{0^n} = 2^{k/2}C_k^{x_k} (I_2 ^{\otimes n}\otimes \bra{x_k}) C_{k-1}^{x_{k-1}} (I_2 ^{\otimes n + 1} \otimes \bra{x_{k-1}})  \cdots C_1^{x_1} (I_2 ^{\otimes n + k - 1} \otimes \bra{x_1}) C_0 \ket{0^n} \otimes \ket{T}^{\otimes k}.
    \end{align}
    where $C_i^{b}$ is a Clifford circuit acting on $n + k - i$ first qubits for all $b\in\set{0, 1}$ and $i\in[n]$.
    In the beginning, $\chi_\delta(\ket{0^n} \otimes \ket{T}^{\otimes k}) = \chi_\delta(\ket{T}^{\otimes k}) $. By Lemma~\ref{lm:T-gadget}, at each measurement, the approximate stabilizer rank does not increase for one output. As approximate stabilizer rank is invariant under Clifford operations, we conclude that for one particular choice of $x_1, \cdots, x_n$
    \begin{align}
        \chi_\delta(\ket{T}^{\otimes k}) &\geq
        \chi_\delta(2^{k/2}C_k^{x_k} (I_2 ^{\otimes n} \otimes \bra{x_k}) C_{k-1}^{x_{k-1}} (I_2 ^{\otimes n + 1}\otimes \bra{x_{k-1}})  \cdots C_1^{x_1} (I ^{\otimes n + k - 1} \otimes \bra{x_1}) C_0 \ket{0^n} \otimes \ket{T}^{\otimes k})\\
        &=\chi_\delta(V\ket{0^n}).
    \end{align}
    
    
    
\end{proof}

\subsection{Concluding the proof of Theorem~\ref{th:lower-bound-stab}}
\label{sec:proof-lb}
Let $c n$ for constant $c > 0$ be the linear term in the statement of Lemma~\ref{lm:t-complex2} and $C$ be the constant in the statement of Lemma~\ref{lm:random-stab-rank}. Given a positive integer $m$, we choose the largest $n$ such that $c n 2^{n/2} \leq m$. By Lemma~\ref{lm:random-stab-rank}, there exists  $n$-qubit state $\ket{\phi}$ with $\chi_{\delta + 4^{-n}}(\ket{\phi}) \geq C\frac{(1-(\delta + 4^{-n})^2)^22^n}{n^2}$  when $n\geq 2\log\frac{1}{1-(\delta+4^{-n})^2} + 9 = O(1)$. According to Lemma~\ref{lm:t-complex2} and our assumption that $cn2^{n/2}\leq m$, there exists a quantum circuit consisting of Clifford gates and at most $m$ number of  $T$ gates such that $\norm{\ket{\phi} |0^{\lambda}\rangle - V(\ket{0^n}|0^\lambda\rangle)} \leq 4^{-n}$ for some $\lambda > 0$. Using Lemma~\ref{lm:stab-rank-circuit}, we  have 
\begin{align}
    \chi_\delta(\ket{T}^{\otimes m})
    &\geq \chi_\delta(V|0^{n + \lambda}\rangle)\\
    &{\geq} \chi_{\delta + 4^{-n}}(\ket{\phi}|0^\lambda\rangle)\\
    &\stackrel{(a)}{\geq} \chi_{\delta + 4^{-n}}(\ket{\phi})\\
    &\geq C\frac{(1-(\delta+4^{-n})^2)^2 2^n}{n^2}.
\end{align}
where $(a)$ follows since $I \otimes \ket{0^\lambda} \bra{0^\lambda}$ maps stabilizer states to scalar multiples of stabilizer states or $0$, and therefore, any approximate decomposition of $\ket{\phi}\ket{0^\lambda}$ into stabilizer states yields a decomposition of $\ket{\phi}$ with the same number of terms and approximation parameter by first applying $I \otimes \ket{0^\lambda} \bra{0^\lambda}$ to each stabilizer term and then tensoring out the last $\lambda$ qubits.
Since $n$ was largest integer such that $c n 2^{n/2} \leq m$, we have $m < \sqrt 2 c (n+1) 2^{n/2}$. Furthermore, we have $2^{n/2} \leq c n 2^{n/2} \leq m$ and therefore $n\leq 2 \log(m)$. Combining these two inequalities, we obtain that $2^{n} > \frac{m}{\sqrt 2 c(2 \log m + 1)}$. Hence, substituting $n$ and $2^n$, we get 
\begin{align}
    C\frac{(1-(\delta+4^{-n})^2)^22^n}{n^2} 
    & > C\frac{(1-(\delta+\frac{4c^4(1+2\log m)}{m^2})^2)^2}{8c^2(1+2 \log m)^2\log^2m} m^2 \\
    &= \Omega\pr*{\frac{(1-\delta^2)^2m^2}{\log^4 m}}.
\end{align}

\begin{remark}
    We proved the result for a fixed $\delta$. However, we highlight the proof works for a sequence  $\set{\delta_m}_{m\geq 1}$ where $\delta_m \to 1$ slowing enough (e.g., $\delta_m = 1 - O(\frac{\log m}{m})$ is enough). More precisely, as long as our constraint $n\geq  2\log\frac{1}{1-(\delta_m+4^{-n})^2} + 9  $ is consistent with our other constraint that $c n 2^{n/2}\leq m$, our argument goes through.
\end{remark}

We conclude this section by stating a corollary of our proof techniques on the existence of quantum states with low circuit complexity but high approximate stabilizer rank.
\begin{corollary}
\label{cor:circuit-rank}
    Let $\delta$, $n$, and $M$ be fixed such that $M = O( (1-\delta^2)^22^nn^{-2})$. Then, there exists a quantum circuit $V$ consisting of $M\poly\log(M) $ gates such that $\chi_\delta(V\ket{0^n}) \geq M$.
\end{corollary}
\begin{proof}
    Let $C$ be the constant in Lemma~\ref{lm:random-stab-rank}. We choose $k$ the smallest integer such that $C(1-\delta^2)^2 \frac{2^k}{k^2} \geq M$. By Lemma~\ref{lm:random-stab-rank}, there exists $k$ qubit state $\ket{\psi}$ such that $\chi_\delta(\ket{\psi}) \geq M$. By \cite[Theorem 3.3]{knill1995approximation}, there exists a quantum circuit $V$ consisting of $O(k2^k) = M\poly\log(M)$ two-qubit gates such that $V\ket{0^k} = \ket{\psi}$. Then, $V\otimes I_2^{\x n-k}$ has the desired properties.
\end{proof}

\section{Approximate stabilizer rank of $t$-designs}
\label{sec:designs}
This section discusses what happens if we replace the Haar measure in Lemma~\ref{lm:random-stab-rank} with a $t$-design. 
Initially, we discuss whether we can expect a bound on the approximate stabilizer rank of $t$-designs. 
The following proposition shatters any such hope for $t=O(1)$.
\begin{prop}
There exist  absolute constants $C_1>0$ and $C_2>0$ such that for all $t$, $n\geq C_2 t^2$, and $1>\epsilon>0$, there exists an $\epsilon$-approximate unitray $t$-design for $n$-qubits where the following holds. If $U$ is sampled according to the $t$-design, $\chi_{\delta}(U\ket{0^n}) \leq 2^{C_1 \log^2(t)(t^4 + t\log(1/\epsilon))}$ with probability one.
\end{prop}
\begin{proof}
    By Theorem~\ref{th:t-design-exsitence}, there exists an $\epsilon$-approximate $t$-design such that if $U$ is sampled from the $t$-design, $U$ with probability one has at most $m = C_1 \log^2(t)(t^4 + t\log(1/\epsilon))$ $T$-gates and arbitrary number of Clifford gates. By Lemma~\ref{lm:stab-rank-circuit}, we have
    \begin{align}
        \chi_\delta(U\ket{0^n}) \leq \chi_\delta(\ket{T}^{\otimes m}) \leq 2^{m}.
    \end{align}
\end{proof}
This proposition together with Lemma~\ref{lm:random-stab-rank} imply that when $t=O(1)$, the approximate stabilizer rank of a $t$-design can range from $O(1)$ to $\Omega(n^{-2}2^n)$. 
However, we provide a lower bound on the approximate rank when $t$ grows with $n$ using a similar approach as in the proof of Lemma~\ref{lm:random-stab-rank}.
\begin{lemma}
    \label{lm:t-design-rank}
    Let $U$ be distributed according to an $\epsilon$-approximate $t$-design for $n$-qubits. Then
    \begin{align}
        \P{\chi_\delta(U\ket{0^n}) \leq M} \leq e^{0.54 n^2M} \frac{\pr*{\frac{M+t-1}{2^n + t - 1}}^t + \epsilon}{(1-\delta^2)^t}
    \end{align}
\end{lemma}
\begin{proof}
    With the same argument as in the proof of Lemma~\ref{lm:random-stab-rank}, we obtain that 
    \begin{align}
        \P{\chi_\delta(U\ket{0^n}) \leq M} \leq e^{0.54 n^2M} \sup_{P\in \Proj{(\C^2)^{\x n}, M}} \P{\norm{PU\ket{0^n}}^2 \geq 1 - \delta^2}.
    \end{align}
    Using Markov's inequality, we further have
    \begin{align}
        \P{\norm{PU\ket{0^n}}^2 \geq 1 - \delta^2}
        &\leq \frac{\E{\norm{PU\ket{0^n}}^{2t}}}{(1-\delta^2)^t}.\label{eq:t-union}
    \end{align}
    We can also write
    \begin{align}
        \E{\norm{PU\ket{0^n}}^{2t}} 
        &= \E{\tr\pr{PU\ketbra{0^n}{0^n} U^\dagger}^t}\\
        &= \E{\tr\pr{P^{\x t}U^{\x t}(\ketbra{0^n}{0^n})^{\x t} (U^\dagger)^{\x t}}}\\
        &= \tr\pr{P^{\x t} M_t^{\nu}((\ketbra{0^n}{0^n})^{\x t})}\\
        &\stackrel{(a)}{\leq} \tr\pr{P^{\x t} M_t^{Haar}((\ketbra{0^n}{0^n})^{\x t})} + \epsilon \\
        &\stackrel{(b)}{\leq} \frac{(M+t -1) \cdots (M+1)M}{(2^n+t-1)\cdots (2^n+1)2^n} + \epsilon\\
        &\leq \pr*{\frac{M+t-1}{2^n + t - 1}}^t + \epsilon\label{eq:t-markov}
    \end{align}
    where $M_t^{\nu}$ is defined in Definition~\ref{def:t-design}, $(a)$ follows since $\nu$ is an $\epsilon$-approximate $t$-design, and $(b)$ follows from Lemma~\ref{lm:haar-expectation}. Substituting the bound from \eqref{eq:t-markov} into \eqref{eq:t-union} completes the proof.
\end{proof}
In the following, we provide an upper bound on how many gates we need to get a state with approximate stabilizer rank $\poly(n)$. Perhaps surprisingly, the bound we obtain using $t$-designs is weaker than what we obtain from the Haar measure in Corollary~\ref{cor:circuit-rank}.
\begin{prop}
\label{th:poly-stab-rank}
    Let $d\geq 1$ and $1>\delta>0$. There exists a quantum circuit $V$ with $O(\log^5(n) n^{3 + 5d + \frac{3\sqrt{(d+1)}}{\sqrt{\log n}} })$ gates and $\chi_{\delta}(V\ket{0^n}) \geq n^{d}$.
\end{prop}
\begin{proof}
    We first apply Lemma~\ref{lm:t-design-rank} to prove that for appropriate choices of $\epsilon$ and $t$, if $\nu$ is an $\epsilon$-approximate $t$-design, then
    \begin{align}
        \P{\chi_\delta(U\ket{0^n}) \leq n^d} < 1.
    \end{align}
    Then, we use Theorem~\ref{th:t-design-exsitence2} to find such $\epsilon$-approximate $t$-design with $O(\log^5(n) n^{3 + 5d + \frac{3\sqrt{(d+1)}}{\sqrt{\log n}} })$ two qubit gates. 
    
    Let $\nu$ be an $\epsilon$-approximate $t$-design for $\epsilon = e^{-n^{d+2}}$ and $t = n^{d + 1}$. By Lemma~\ref{lm:t-design-rank},
    \begin{align}
        \P{\chi_\delta(U\ket{0^n}) \leq n^d} 
        &\leq  e^{0.54 n^{d+2}} \frac{\pr*{\frac{n^d+t-1}{2^n + t - 1}}^t + \epsilon}{(1-\delta^2)^t}\\
        & =  e^{0.54 n^{d+2}} \frac{\pr*{\frac{n^d+n^{d+1}-1}{2^n + n^{d+1} - 1}}^{n^{d+1}}} + e^{-n^{d+2}}{(1-\delta^2)^{n^{d+1}}}\\
        &\leq  e^{0.54 n^{d+2}} \frac{\pr{2n^{d+1}}^{n^{d+1}}2^{-n^{d+2}} + e^{-n^{d+2}}}{(1-\delta^2)^{n^{d+1}}}\\
        &\leq e^{-\Omega(n^{d+2})},
    \end{align}
    as claimed. Next, it is a direct calculation for our choice of $t$ and $\epsilon$
    \begin{align}
        C n\ln^5(t) t^{4+3\frac{1}{\sqrt{\log(t)}}}(2nt + \log(1/\epsilon)) =  O(\log^5(n) n^{3 + 5d + \frac{3\sqrt{(d+1)}}{\sqrt{\log n}} })
    \end{align}
    which is the number of two-qubits gate in an $\epsilon$-approximate $t$-design according to Theorem~\ref{th:t-design-exsitence2}.
\end{proof}

We finally comment on the possibility of using $t$-design to improve the lower bound on the approximate stabilizer of $\ket{T}^{\x m}$. Using Lemma~\ref{lm:t-design-rank}, we should take $t = \Omega(Mn)$ for a state with approximate stabilizer state $M$. Assuming that Conjecture~\ref{con:converse-t-design} is true, we need at least $\Omega(Mn)$ $T$-gates to construct such $t$-design. Therefore, using our approach, one cannot expect to get a lower bound better than linear on $\chi_\delta(\ket{T}^{\x m})$.

\section{Acknowledgement}

S. M. and M. T. acknowledge funding provided by NSF CCF-2013062. S. M. acknowledges funding by the Institute for Quantum Information and Matter, an NSF Physics Frontiers Center (NSF Grant PHY-1733907). S. M. and M. T. are grateful to Ryan Williams, Ulysse Chabaud, and Arsalan Motamedi for their insightful conversations. We thank the anonymous referee who helped us spot an error in Theorem 1.6. (now fixed).

\appendix
\section{Stabilizerness measures}
\label{a:measures}

\begin{table}[]
    \centering
\begin{tabular}{|c|c|c|c|c|c|}
\hline
     &  $\xi(\ket{\phi})$ & $\chi(\ket{\phi})$ & $\chi_\delta(\phi)$ & $\norm{\ket{\phi}}_{U_3}$ \\ \hline
     $F(\ket{\phi})$  & \makecell{$\zeta(\ket{\phi}) = \sup_{\ket{\omega}} \frac{\abs{\braket{\phi}{\omega}}^2}{F(\ket{\omega})}$\\ \cite[Theorem 4]{Bravyi_2019} }   & \makecell{Proposition~\ref{prop:f-chi}\\ See also \cite{Labib2022stabilizerrank}}& Open question& \makecell{Remark~\ref{rem:gowers}\\\cite{samorodnitsky2007low}}  \\ \hline
     $\xi(\ket{\phi})$&\cellcolor{gray} &Open question&\cite[Theorem 1]{Bravyi_2019}& Open question\\ \hline
     $\chi(\ket{\phi})$& Open question&\cellcolor{gray}&\makecell{Lemma~\ref{lem:gap}\\\cite[Lemma~3.5]{Lovitz2022newtechniques}}&Open question\\ \hline
\end{tabular}
    \caption{Known relations among measures of stabilizerness.}
    \label{tab:stab-measures}
\end{table}
Approximate and exact stabilizer ranks measure the closeness of a quantum state to the set of all stabilizer states. Other measures of stabilizerness have been studied in the literature that might be mathematically more tractable but are not operationally as relevant as stabilizer ranks. In this section, we review some of these measures and also suggest a novel one based on the Gowers norm \cite{Gowers_2001} (See Definition~\ref{def:stab-measures}). Gowers norms are extensively studied in the context of higher-order Fourier analysis and have found several applications in different areas of mathematics and theoretical computer science \cite{tao2012higher, hatami2019higher}. In particular, for a polynomial $P:\F^n\to \F$ of degree $d$, the Gowers norm of degree $d$ of $(-1)^P$ is 1. Since stabilizer states are defined in terms of quadratic phases, the useful properties of Gowers norm could be exploited to study stabilizerness of arbitrary quantum states  (See Remark~\ref{rem:gowers}). Since some of these measures are easier to deal with, it is interesting to ask if a bound on one measure implies any bound on other measures. As an example, we prove a relation between stabilizer fidelity and stabilizer rank (Proposition~\ref{prop:f-chi}), which immediately implies a linear lower-bound on $\chi(\ket T^{\x n}) $ using the results of \cite{Bravyi_2019} (Corollary~\ref{cor:t-exact}). We summarize other relations in Table~\ref{tab:stab-measures} that exist in literature.
\begin{definition} [Stabilizerness measures]
\label{def:stab-measures}
    Let $\ket{\phi} = \frac{1}{2^{n/2}}\sum_{x\in\F^n} f(x) \ket{x}$ be an $n$-qubit state. We consider the following quantities.
    \begin{enumerate}
    \item Stabilizer fidelity: $F(\ket{\phi}) = \max_{\ket{s} \in \Stab_n} \abs{\braket{s}{\phi}}^2$
    \item Stabilizer extent: 
    \begin{align}
        \xi(\phi) = \inf\set{\norm{c}_1^2: c \in \C^M: \exists \ket{s_1}, \cdots, \ket{s_M} \in \Stab_n: \ket{\phi} = \sum_{i=1}^M c_i \ket{s_i}}.
    \end{align}
    \item Stabilizer rank $\chi(\ket{\phi})$ and approximate stabilizer rank $\chi_\delta(\ket{\phi})$ as defined in Section~\ref{sec:stab-form}.
    \item Gowers norm:
    \begin{multline}
        \norm{\ket{\phi}}_{U^3}^8 = \frac{1}{16^n} \sum_{x, h_1, h_2, h_3\in\F^n} f(x) \overline{f(x+h_1) f(x+h_2) f(x+h_3)} \\
        \times f(x+h_1+h_2)f(x+h_1 + h_3) f(x+h_2+h_3) \overline{f(x+h_1+h_2+h_3)}.
    \end{multline}

\end{enumerate}
\end{definition}
\begin{remark}
\label{rem:gowers}
    We believe that Gowers norm $\norm{\ket{\phi}}_{U^3}$ is a relevant measure of stabilizerness because of its relation with the overlap of a function with quadratic phase functions \cite[Theorem~5.3]{hatami2019higher}. In particular, let $\norm{\ket{\phi}}_{U^3} = \delta$. Then, the direct part of \cite[Theorem~5.3]{hatami2019higher} implies that for any stabilizer state $\ket{s}$ such that $\braket{s}{x} \neq 0$ for all $x\in \F^n$, we have $\abs{\braket{s}{\phi}} \leq \delta$. Furthermore, the converse part of \cite[Theorem~5.3]{hatami2019higher} implies that there exists a stabilizer state $\ket{s}$ such that $\abs{\braket{s}{\phi}} \geq 2^{-c\log^4\frac{1}{\delta}}$ for a universal constant $c>0$. While the direct and converse theorems do not match, $\norm{\ket{\phi}}_{U^3}$ has a closed expression unlike $F(\ket{\phi})$.
\end{remark}

In the following proposition, we, establish a relation between $F(\ket{\phi})$ and $\chi(\ket{\phi})$. We closely follow the proof approach of \cite{Labib2022stabilizerrank}, but we manage to avoid any tools from the higher-order Fourier analysis. 

\begin{prop}
\label{prop:f-chi}
For any $n$-qubit state $\ket{\phi} = \frac{1}{\sqrt{2^n}} \sum_{x\in\F^n}f(x)\ket{x}$, we have 
\begin{align}
    \chi(\ket{\phi}) \geq \frac{2}{3} \log\pr*{\frac{\alpha^2}{\beta\sqrt{F(\ket{\phi})}}}
\end{align}
where
\begin{align}
    \alpha \eqdef \min_{x\in\F^n} \abs{f(x)}\\
    \beta \eqdef \max_{x\in\F^n} \abs{f(x)}.
\end{align}
\end{prop}

 Before starting the proof, we briefly review the setup for (classical) Fourier analysis here. Let $G$ be any finite Abelian group. A character for $G$ is a function $\gamma:G \to \C$ such that $\gamma(g + h) = \gamma(g) \gamma(h)$. We denote the set of all characters by $\widehat G$, which is finite and of the same cardinality as $G$. For any function  $f:G\to \C$, there exists a unique function $\widehat{f}:\widehat{G} \to \C$ such that $f= \sum_{\gamma \in \widehat{G}}\widehat{f}(\gamma) \gamma$. Using Parseval equality and Cauchy Schwartz inequality, we have
\begin{align}
     \sum_{\gamma \in \widehat{G}} \abs{\widehat{f}(\gamma)} \leq \sqrt{\abs{G}} \max_{g\in G} \abs{f(g)}.
\end{align}
Let $G_1$ and $G_2$ be two finite Abelian groups. We have $\widehat{G_1 \times G_2} = \set{\gamma_1\gamma_2: \gamma_1\in \widehat{G_1}, \gamma_2\in\widehat{G_2}}$. Finally, consider the group $G=\set{\pm 1, \pm i}$ under multiplication. Then, $\widehat{G} \cong	
 \Z_4$ and the characters are $x \mapsto x^i$ for $i\in \Z_4$.
\begin{proof}[Proof of Proposition \ref{prop:f-chi}]
We define the set of functions $\F^n \to \C$ of the form $x\mapsto i^{\ell(x)}(-1)^{Q(x)}$ for $\ell$ linear and $Q$ quadratic as $\mathcal{Q}$, which is closed under function multiplication.
Let $\chi_\delta(\ket{\phi}) = M$, i.e., there exist $c_1, \cdots, c_M\in \C$, $A_1, \cdots, A_M \subset \F^n$ affine subspaces, $Q_1, \cdots, Q_M\in \mathcal{Q}$ such that 
\begin{align}
    f = \sum_{i=1}^M c_i Q_i 1_{A_i}.
\end{align}
We recursively construct an affine subspace $U\subset \F^n$ with dimension at least $n - M$ such that $1_{A_i}$ is constant on $U$ for all $i$. \footnote{This is a minor improvement to  Claim~3.3 in \cite{Peleg2022lowerbounds}.}. Start with $U_0 = \F^n$. For each $i\in[M]$ let  $U_{i-1}$ be given such that $1_{A_1}, \cdots, 1_{A_{i-1}}$ are constant on $U_{i-1}$. If $U_{i-1} \subset A_i$, then take $U_i = U_{i-1}$ and $1_{A_i}$ is one on $U_i$. Otherwise, there is a subspace $U_i$ of $U_{i-1}$ with codimension one such that $U_{i-1}\cap A_i = \empty$, i.e., $\one_{A_i}$ is zero on $U_i$ (and since $U_{i} \subset U_{i-1}$,  $1_{A_1}, \cdots, 1_{A_{i-1}}$ are constant on $U_i$ too). We take $U = U_M$, which has a dimension at least $n - M$ because the dimension of $U_0$ is $n$, and in each step, the dimension is decreased at most by one. 

Upon defining $S \eqdef \set{i\in [M]: 1_{A_i}|_U = 1}$, we have
\begin{align}
    f|_U = \sum_{i\in S} c_i Q_i|_U.
\end{align}

Define  $G = \set{\pm 1, \pm i}$, which is a group under multiplication. We also define
\begin{align}
    &h:U \to G^S\quad x\mapsto (Q_i(x))_{i\in S}\\
    &\Gamma: G^S \to \C \quad y \mapsto \begin{cases}\sum_{i\in S}c_i y_i\quad&y\in \textnormal{Im}(h)\\ 0 \quad \textnormal{Otherwise} \end{cases}
\end{align}
for which
\begin{align}
    f(x) = \Gamma(h(x)) ~\forall x \in U.
\end{align}

Using our discussion before the proof about Fourier analysis of $G^S$, we can write
\begin{align}
    \Gamma = \sum_{\gamma \in \Z_4^S} \widehat{\Gamma}(\gamma) \gamma.
\end{align}
and 
\begin{align}
    f(x) = \sum_{\gamma \in \Z_4^S} \widehat{\Gamma}(\gamma) Q_\gamma(x) ~\forall x\in U
\end{align}
where $Q_\gamma \eqdef \prod_{i\in S} Q_i^{\gamma_i} \in \mathcal{Q}$. Therefore, 
\begin{align}
    \label{eq:fourier-ex}
     \frac{1}{\abs{U}} \sum_{x\in U} \abs{f(x)}^2 = \sum_{i\in S} \widehat{\Gamma}(\gamma) \frac{1}{\abs{U}}\sum_{x\in U} Q_\gamma(x) \overline{f(x)}.
\end{align}

We consider the stabilizer state 
\begin{align}
     \ket{s_\gamma} \eqdef \frac{1}{\sqrt{\abs{U}}} \sum_{x\in U} Q_\gamma(x) \ket{x}.
\end{align}
It holds that
\begin{align}
\label{eq:stab-inner}
    \braket{\phi}{s_\gamma} = \frac{1}{\sqrt{\abs{U}2^n}}\sum_{x\in U}Q_\gamma(x) \overline{f(x)}.
\end{align}
Combining \eqref{eq:fourier-ex} and \eqref{eq:stab-inner}, we obtain
\begin{align}
    \frac{1}{\abs{U}} \sum_{x} \abs{f(x)}^2 
    &= \sum_{i\in S} \widehat{\Gamma}(\gamma) \sqrt{\frac{2^n}{|U|}}  \braket{\phi}{s_\gamma}\\
    &\leq \pr{\sum_{\gamma} \abs{\widehat{\Gamma}(\gamma)}}\sqrt{\frac{2^nF(\ket{\phi})}{\abs{U}}}\\
    &\leq 2^{|S|} \max_{y\in G^S}\abs{\Gamma(y)} \sqrt{\frac{2^nF(\ket{\phi})}{\abs{U}}}\\
    &\leq 2^{|S|} \beta \sqrt{\frac{2^nF(\ket{\phi})}{\abs{U}}}\\
    & \leq 2^{M} \beta \sqrt{\frac{F(\ket{\phi})}{2^{M}}}
\end{align}
Combining the above inequality with $\frac{1}{\abs{U}} \sum_{x} \abs{f(x)}^2  \geq \alpha^2$ completes the proof.
\end{proof}
Proposition~\ref{prop:f-chi} immediately implies a linear lower bound on exact stabilizer rank of $\ket{T}^{\x m}$ using $F(\ket{T}^{\x m}) = \cos(\frac{\pi}{8}) ^{2m}$ \cite{Bravyi_2019}.

\begin{corollary}
\label{cor:t-exact}
We have
\begin{align}
    \chi(\ket{T}^{\x m}) \geq \frac{2}{3} \log\pr*{\frac{1}{\cos(\frac{\pi}{8})}} m > 0.076 m.
\end{align}
\end{corollary}

The next lemma shows that we cannot, in general, control the gap between approximate and exact rank.
\begin{lemma}
\label{lem:gap}
    Let $1>\delta >0$, $n$, and $M$ be given such that $1 \leq M\leq 2^n$. Then, there are two $n$-qubit quantum states $\ket{\phi_1}$ and $\ket{\phi_2}$ such that
    \begin{align}
        \chi(\ket{\phi_1}) &= \chi(\ket{\phi_2}) = \Theta(M)\\
        \chi_\delta(\ket{\phi_1}) &= \Omega\pr*{\frac{M}{\log^2M}}\\
        \chi_\delta(\ket{\phi_2}) &= O(1).
    \end{align}
\end{lemma}

\begin{proof}
    Let $k$ be the smallest integer such that $2^k\geq M$. By Lemma~\ref{lm:random-stab-rank}, there exists a $k$ qubit state $\ket{\psi_1}$ such that 
    \begin{align}
        \chi_\delta(\ket{\psi_1}) = \Omega(\frac{2^k}{k^2}).
    \end{align}
    Furthermore, by a probabilistic argument, we can further assume that 
    \begin{align}
        \chi_\delta(\ket{\psi_1}) = 2^k.
    \end{align}
    Taking $\ket{\phi_1} = \ket{\psi_1} \x \ket{0}^{\x n - k}$ has the desired property.
Next, using the same argument and \cite[Lemma~3.5]{Lovitz2022newtechniques} implies the existence of $\ket{\phi_2}$. 
\end{proof}

\bibliographystyle{alphaurl}
\newcommand{\etalchar}[1]{$^{#1}$}

\end{document}